\documentclass[aps,prl,reprint,superscriptaddress]{revtex4-2}
\usepackage{amsthm}
\usepackage{amsmath,amssymb}
\usepackage{graphicx}
\usepackage{dcolumn}
\usepackage{bm}
\usepackage{bbm}

\usepackage{xcolor}
\usepackage{siunitx}
\usepackage{braket}
\usepackage{booktabs}
\usepackage{tikz}
\usepackage{xr-hyper}
\usepackage[colorlinks=true,linkcolor=blue,citecolor=blue,urlcolor=blue]{hyperref}

\usetikzlibrary{matrix,calc}

\newtheorem{theorem}{Theorem}
\DeclareMathOperator{\Tr}{Tr}
\makeatletter
\newcommand*{\rom}[1]{\expandafter\@slowromancap\romannumeral #1@}
\makeatother
\newcommand{\id}{\mathbbm{1}} 

\begin{document}

\title{Three Hamiltonians are Sufficient for Unitary $k$-Design in Temporal Ensemble}

\author{Yi-Neng Zhou$^{\dagger}$}
\email{zhouyn.physics@gmail.com}
\affiliation{Department of Theoretical Physics, University of Geneva, 24 quai Ernest-Ansermet, 1211 Genève 4, Suisse}

\author{Tian-Gang Zhou}
\thanks{These authors contributed equally to this work.}
\affiliation{DQMP, University of Geneva, 24 quai Ernest-Ansermet, 1211 Genève 4, Suisse}

\author{Julian Sonner}
\email{julian.sonner@unige.ch}
\affiliation{Department of Theoretical Physics, University of Geneva, 24 quai Ernest-Ansermet, 1211 Genève 4, Suisse}

\begin{abstract}


Unitary $k$-designs are central to quantum information and quantum many-body physics, as they provide efficient proxies for Haar-random dynamics. We study how chaotic Hamiltonian evolution can generate unitary $k$-designs using the frame potential (FP). Unlike standard approaches based on independently sampled Hamiltonians or fine-tuned evolution times, we consider a quenched temporal ensemble in which the Hamiltonians are sampled once and then held fixed, while randomness enters only through the evolution times.
We compare a two-step protocol (2SP), with evolution under two fixed Hamiltonians, to a three-step protocol (3SP) with one additional quench, with three evolution times sampled independently. Analytically and numerically, we show that the 2SP fails to realize general unitary $k$-designs, whereas the 3SP realizes them for arbitrary $k$. The random phases in the 3SP impose stronger index-matching constraints in the FP, eliminating the independent permutation degrees of freedom that still remain in the 2SP. Even with imperfect time averaging, the 3SP reaches a given accuracy within a parametrically shorter time window. These results hold universally across different symmetry classes and persist in both all-to-all and local models with random interactions.

\end{abstract}

\maketitle

\begin{figure}[t] 
    \centering \includegraphics[width=0.95\linewidth]{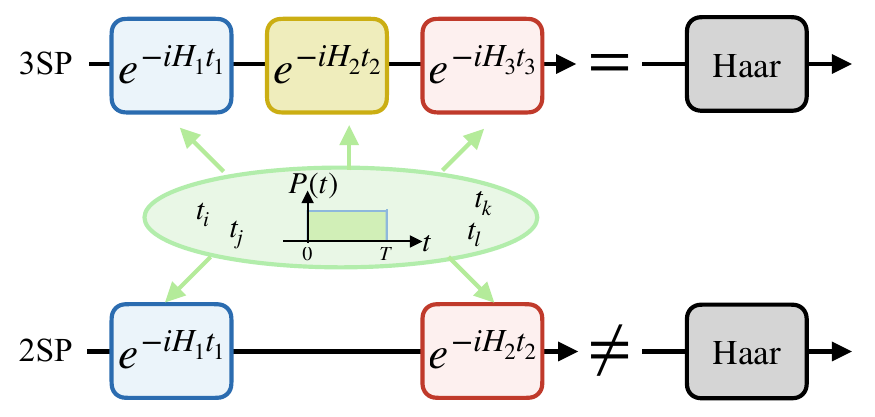} 
    \caption{
    Random ensembles generated by the two-step (2SP) and three-step (3SP) protocols with fixed Hamiltonians $H_1$, $H_2$, and $H_3$. Each box denotes evolution under the corresponding fixed Hamiltonian for a duration $t_1$, $t_2$, or $t_3$, where the times are drawn independently from the same distribution $P(t)$ illustrated by the curve within the central green circle. Theorem~\ref{thm:perfect_filter} shows that 2SP cannot realize a general unitary $k$-design, whereas 3SP can, and is also random matrix symmetry-class invariant.}
\label{temporary_ensemble_fig}
\end{figure}

    {\em \color{blue} Introduction---} Randomness is a defining feature of chaotic quantum dynamics. At sufficiently late times, complex many-body evolution is often expected to exhibit statistical properties resembling those of Haar-random unitaries, making Haar random a universal benchmark for quantum chaos, thermalization, and information scrambling \cite{Roberts_2017,Cotler_2017,Brown_2018,Liu_2018,Page_1993}. Also, random unitaries have become powerful tools throughout quantum information, with applications ranging from randomized benchmarking and quantum tomography \cite{Huang_2020,PhysRevLett.125.200501,Elben_2022,ODonnell2015EfficientQT} to randomized measurements \cite{Emerson_2005,Knill_2008}, quantum computation\cite{bernstein1993quantum}, and cryptography \cite{PortmannRenner2022RMP}. Haar averages are also analytically powerful: invariance and concentration often turn otherwise intractable averages into closed forms and controlled large-$D$ asymptotics \cite{Collins_2006,Hayden_2006,Popescu_2006}. Haar randomness thus serves as a universal benchmark for maximally random dynamics \cite{Page_1993}, and unitary $k$-designs reproduce this benchmark up to the first $k$ moments \cite{Dankert2009PRA,ambainis2007quantumtdesignstwiseindependence}.

This raises a basic question at the interface of quantum chaos and quantum information:
How simply can Haar-like randomness emerge from deterministic many-body dynamics\cite{zhou2025realizingunitarykdesignssingle}?
Exact constructions typically rely on deep local random circuits \cite{RandomUnitary2003, Dankert2009PRA, Brand2016, Brand_o_2016, Brand_o_2021, Schuster:2024ajb}, Brownian chaotic Hamiltonians with frequently modulated random parameters \cite{Jian2022LinearGO, Tiutiakina2023FramePO, Onorati_2017, Brand_o_2021, Guo_2024}, or Floquet schemes with many stroboscopic layers \cite{RandomUnitary2003,MerkelRiofrioFlammiaDeutsch2010PRA, HoChoi2022PRL,IppolitiHo2022Quantum, FarshiEtAl2022JMP}. In practice, these approaches often require sampling many independent Hamiltonian realizations or applying many layers of independently chosen random gates, which could be experimentally costly \footnote{Temporal randomization can be particularly advantageous on platforms where changing a Hamiltonian realization requires reprogramming and recalibrating many control parameters. Once a small set of Hamiltonians is calibrated, different ensemble realizations can instead be generated simply by varying their evolution times, reducing the experimental control overhead. If Hamiltonian parameters or random gates can be varied at negligible additional cost, Hamiltonian randomization may be equally convenient or even preferable.}. This motivates a more fundamental question: can the intrinsic randomness of chaotic Hamiltonian dynamics itself generate unitary designs with only minimal control? Such a mechanism would also be experimentally appealing, especially for analog quantum simulators, where evolution under a few fixed Hamiltonians is natural and considerably simpler than implementing deep random circuits or repeatedly reconfiguring the Hamiltonian \cite{zhou2025realizingunitarykdesignssingle}.

In this Letter, we study how to realize unitary $k$-designs using a minimal quenched temporal ensemble generated by time evolution under a small number of fixed chaotic Hamiltonians, where randomness enters only through sampled evolution times \cite{PhysRevX.14.041059,h6xy-zpx4,PhysRevLett.131.250401}.
The basic building block is a two-step protocol (2SP), in which the unitary is
$V(t_1,t_2)=e^{-iH_2 t_2}e^{-iH_1 t_1}$ with fixed $(H_1, H_2)$ chosen from the same random Hamiltonian ensemble and independently sampled $(t_1,t_2)$, as shown in Fig.~\ref{temporary_ensemble_fig}.
We express its frame potential (FP), which quantifies the deviation from Haar randomness \cite{Roberts_2017,Cotler_2017,Mele_2024}, in terms of inverse participation ratios (IPR)~\cite{Wegner1980, KramerMacKinnon1993, Evers_2000, Evers_2008} of the eigenbasis overlap matrix between $H_1$ and $H_2$, revealing that design formation is governed by eigenvector-overlap statistics. Analytically and numerically, we find that the 2SP fails to form a unitary $k$-design for $k>1$, whereas the 3SP with one additional quench achieves a unitary $k$-design for general $k$. The underlying mechanism admits a simple explanation: time filtering enforces energy-index matching in the FP, leaving permutation degrees of freedom. When the overlap matrices carry effectively random phases, the 3SP produces strong phase cancellations so that only a single permutation survives, yielding the Haar value.
In contrast, the 2SP retains two independent permutation freedoms and therefore a parametrically larger FP.
We also quantify finite-$T$ imperfect-filter corrections, showing they are parametrically smaller for 3SP. Surprisingly, this result is highly universal and holds for Hamiltonians drawn from the Gaussian Orthogonal, Unitary, and Symplectic Ensembles (GOE,GUE,GSE)~\cite{Collins2002,Mehta2004,Collins_2006,Anderson2009}. We verify them numerically, including various local and all-to-all randomly interacting model.

{\em \color{blue} Quenched temporal ensemble in 2SP---} We consider a unitary ensemble generated by time evolution under a fixed sequence of chaotic Hamiltonians, where the only randomness comes from the evolution times. The simplest case is the 2SP with two fixed Hamiltonians $H_1$ and $H_2$
\begin{equation}
V(t_1,t_2)=e^{-iH_2 t_2}e^{-iH_1 t_1},
\label{V_definition}
\end{equation}
for $t_1,t_2\in[0,T]$.
Sampling $t_1$ and $t_2$ independently from a distribution $P(t)$ on $[0,T]$ defines a unitary ensemble $\nu$.
Its $k$-th frame potential (FP) \cite{Roberts_2017,Cotler_2017,Mele_2024} is
\begin{equation}
F^{(k)}_{\nu}=\mathbb{E}_{V_1,V_2\sim\nu}\,|\Tr(V_1^\dagger V_2)|^{2k},
\end{equation}
with $k\in\mathbb{N}$. An ensemble forms a unitary $k$-design if and only if $F^{(k)}_{\nu}$ equals the Haar value $k!$. For 2SP, inserting Eq.~\eqref{V_definition} gives
\begin{equation}
F^{(k)}_{\mathrm{2SP}}
=\mathbb{E}_{\oplus_{j=1}^2\{t_j,t_j'\}}
\Bigl|\Tr\!\left(e^{iH_1 \Delta t_1}e^{iH_2 \Delta t_2}\right)\Bigr|^{2k},
\label{eq:Fk_trace_shifted}
\end{equation}
where $\Delta t_j=t_j-t_j'$. The temporal ensemble average above is defined as \footnote{Our definition of the temporal ensemble here differs from the previous one \cite{PhysRevX.14.041059,h6xy-zpx4, PhysRevLett.131.250401} in two respects: it involves two fixed Hamiltonians rather than one, and the evolution time $t$ is sampled from a general distribution $P(t)$. The latter may be viewed as an additional source of classical randomness. When $P(t)$ is uniform on $[0, T]$, our definition reduces exactly to the previous one, and we adopt this choice below for simplicity.} $$\mathbb{E}_{\oplus_{j=1}^2\{t_j,t_j'\}} [\,\cdot\,]
\equiv
\int_0^T \prod_{j=1}^2 dt_j\,dt_j'\,
P(t_j) P(t_j')  [\,\cdot\,].$$ 

To clarify what \eqref{eq:Fk_trace_shifted} probes, we first consider $k=1$. Expanding in the eigenbases
$H_1|E_m\rangle=E_m|E_m\rangle$ and $H_2|\epsilon_n\rangle=\epsilon_n|\epsilon_n\rangle$, the time integrals factorize, it gives
\begin{equation}
F^{(1)}_{\mathrm{2SP}}=\sum_{m,n,m',n'=1}^D |U^{(\bm{1})}_{mn}|^2\, |U^{(\bm{1})}_{m'n'}|^2\,
I_T(E_{m,m'})\;
I_T(\epsilon_{n,n'}),
\label{eq:F1_with_filter_general}
\end{equation}
with $U^{(\bm{1})}_{mn}\equiv \braket{E_m|\epsilon_n}$, $E_{m,m'}\equiv E_m-E_{m'}$, and $\epsilon_{n,n'}\equiv \epsilon_n-\epsilon_{n'}$.
The time-filter function is $I_T(\Delta E)\equiv
\left|\int_0^T\!dt\,P(t)\,e^{i\Delta E t}\right|^2$.
For concreteness, we take the uniform distribution $P(t)$
which is only nonzero on the interval $t\in[0,T]$. For a discrete, nondegenerate spectrum,
\begin{equation}
I_T(E_{m,m'})
=\operatorname{sinc}\left(E_{m,m'} T/2\right)^2,
\label{eq:IT_sinc2}
\end{equation}
with $\lim_{T\to\infty} I_T(E_{m,m'}) =
\delta_{mm'}$ and similarly for $\epsilon_{m,m'}$. Applying this to Eq.~\eqref{eq:F1_with_filter_general} yields 
$\lim_{T\to\infty} F^{(1)}_{\mathrm{2SP}}(T)
=\sum_{m,n=1}^D |U^{(\bm{1})}_{mn}|^4$,
i.e., the IPR of the change-of-basis matrix between the eigenbases of $H_1$ and $H_2$.
In particular, for a perfectly flat overlap matrix (FOM), $|\langle E_m|\epsilon_n\rangle|^2=1/D$, one obtains
$F^{(1)}_{\mathrm{2SP}}(T \to \infty)=1$, matching the Haar value at $k=1$.

We now analyze the general $k$-th FP of the quenched temporal ensemble.
Expanding the trace in Eq.~\eqref{eq:Fk_trace_shifted} in the eigenbases of $H_1$ and $H_2$, the $k$-th FP becomes
\begin{equation}
    \begin{split}
        F^{(k)}_{\mathrm{2SP}}=& 
        \sum_{m_a,n_a, m_a',n_a'=1}^D \prod_{a=1}^k \left[\Bigl|U^{(\bm{1})}_{m_a n_a} \Bigr|^2\, \Bigl|U^{(\bm{1})}_{m_a' n_a'} \Bigr|^2 \right] \\ &\mathbb{E}_{\oplus_{j=1}^2\{t_j,t_j'\}} e^{i \sum_{a=1}^k\left( \Delta t_1  E_{m_a,m'_a} + \Delta t_2  \epsilon_{n_a,n'_a}\right)}.
    \end{split}
\label{eq:FP_two_time_expanded}
\end{equation}
For long times, when $T$ exceeds the inverse energy-level spacings (the Heisenberg time), the time integral enforces the additive energy constraints
$\sum_{a=1}^k (E_{m_a}-E_{m_a'})=0$ and $
\sum_{a=1}^k (\epsilon_{n_a}-\epsilon_{n_a'})=0$. Here, the $H_1$ and $H_2$-sector constraints decouple. Assuming the absence of spectral resonances, the surviving terms correspond to independent permutations of the $k$ indices in each sector: $m_a' = m_{\pi(a)}$ and $n_a' = n_{\sigma(a)}$. This yields
\begin{equation}
    \begin{split}
    \label{eq:2step_FP_perfect}
        &F^{(k)}_{\mathrm{2SP}}= 
        \sum_{\pi,\sigma \in S_k}^D\sum_{m_a,n_a=1}^D  \prod_{a=1}^k \Bigl|U^{(\bm{1})}_{m_a n_a} \Bigr|^2\,\Bigl|U^{(\bm{1})}_{m_{\pi(a)} n_{\sigma(a)}} \Bigr|^2 .
    \end{split}
\end{equation}
In the FOM case, the calculation yields $F^{(k)}_{\mathrm{2SP}}=(k!)^{2}$, since both $\pi,\sigma$ permutations contribute. For $k>1$, this is parametrically larger than the Haar value $F^{(k)}_{\mathrm{Haar}}=k!$. Therefore, the 2SP cannot realize higher-order $k$-designs. We thus introduce one additional quench and will show below that moving to the 3SP already suffices to realize a general unitary $k$-design.

{\em \color{blue} 3SP in general $k$---} 
Now we consider the 3SP. Its time evolution is given by
\begin{equation}
V(t_1,t_2,t_3)=e^{-iH_3 t_3}e^{-iH_2 t_2}e^{-iH_1 t_1},
\end{equation}
for $t_1,t_2,t_3\in[0,T]$.
Here, $H_1, H_2, H_3$ are drawn independently and then held fixed. Its $k$-th FP is then
\begin{equation}
    \begin{split}
        F^{(k)}_{\mathrm{3SP}}
        =&\mathbb{E}_{\oplus_{j=1}^3\{t_j,t_j'\}} 
        \Bigl|\Tr\!\left(e^{iH_1 \Delta t_1}
        e^{iH_2 t_2}
        e^{iH_3 \Delta t_3}
        e^{-iH_2 t_2'}\right)\Bigr|^{2k}.
    \end{split}
\label{eq:FP_three_time_def}
\end{equation}
We expand the trace in the eigenbases of $H_1,H_2,H_3$, denoted by
$\{\ket{E_m}\}$, $\{\ket{\epsilon_p}\}$, and $\{\ket{\eta_g}\}$. The result involves the basis overlapping term
$A_{mpgf}
=
U^{(1)}_{mp}\,U^{(2)}_{pg}\,U^{(2)*}_{fg}\,U^{(1)*}_{mf}$ and the phase factor
$$
B_{mpgf}(t_j,t_j')
=
\exp\!\Bigl[
i\Delta t_1\,E_m
+i t_2\,\epsilon_p
-i t_2'\,\epsilon_f
+i\Delta t_3\,\eta_g
\Bigr],$$ where $j\in\{1,2,3\}$ and the overlap matrix is defined by $U^{(1)}_{mp}\equiv \braket{E_m|\epsilon_p}$ and $U^{(2)}_{pg}\equiv \braket{\epsilon_p|\eta_g}$.
The resulting $k$-th FP can be written as
\begin{equation}
\begin{split}
F^{(k)}_{\mathrm{3SP}}
=&\mathbb{E}_{\oplus_{j=1}^3\{t_j,t_j'\}}
 \sum_{\{\text{indexes}=1\}}^D  
\prod_{a=1}^k \Big( B_{m_ap_ag_a f_a}(t_j,t_j') \\
& B^*_{m_a'p_a'g_a' f_a'}(t_j,t_j') A_{m_ap_ag_a f_a} A^{*}_{m_a' p_a' g_a' f_a'} \Big) ,
\end{split}
\label{eq:FP_three_time_expanded_k}
\end{equation}
where $\text{indexes} = \{m_a,p_a,f_a,g_a,m_a',p_a',f_a',g_a'\}$.

In the perfect time-filtering limit, the energy constraints are
$\sum_{a=1}^k (E_{m_a}- E_{m_a'})=0$ and similarly for $\epsilon_{p_a}$, $\epsilon_{f_a}$ and $\eta_{g_a}$. A naive counting would then suggest a factor $(k!)^4$ in the FP, corresponding to $4$ independent permutations $m_a' = m_{\pi(a)},\ 
p_a' = p_{\sigma(a)},\ 
f_a' = f_{\upsilon(a)},g_a'=g_{\tau(a)}$. 
\begin{equation}
\label{eq:3step_FP_perfect}
\begin{split}
    F^{(k)}_{\mathrm{3SP}}
=&\sum_{\pi,\sigma,\tau,\upsilon\in S_k}\sum_{\{m,p,f,g\}=1}^D
\prod_{a=1}^k
\Bigl(U^{(1)}_{m_a p_a}\,U^{(2)}_{p_a g_a}\,U^{*(2)}_{f_a g_a}\,U^{*(1)}_{m_a f_a}\Bigr) \\
& \Bigl(U^{(1)}_{m_{\pi(a)} p_{\sigma(a)}}\,U^{(2)}_{p_{\sigma(a)} g_{\tau(a)}}\,U^{*(2)}_{f_{\upsilon(a)} g_{\tau(a)}}\,U^{*(1)}_{m_{\pi(a)} f_{\upsilon(a)}}\Bigr)^{*},
\end{split}
\end{equation} 

However, these permutations are not independent. The overlap structure and phase cancellation enforce adjacent-index matching across conjugate pairings. To see this, consider the factors $\langle E_{m_a}|\epsilon_{p_a}\rangle \langle \epsilon_{f_a}|E_{m_a}\rangle$
in $A_{m_a p_a g_a f_a}$, and pair them with the conjugate factors $\langle E_{m_b'}|\epsilon_{f_b'}\rangle^{*}
\langle \epsilon_{p_b'}|E_{m_b'}\rangle^{*}$
in $A^*_{m_b' p_b' g_b' f_b'}$. Choosing $b=\pi^{-1}(a)$, the index becomes $m_b'=m_a$, $p_b'=p_{\sigma\pi^{-1}(a)}$ and
$f_b'=f_{\upsilon\pi^{-1}(a)}$.
For the phases to survive, we must have $p_b'=p_a$ and $f_b'=f_a$. Otherwise, the corresponding overlaps vanish after averaging. This forces
$\pi=\sigma=\upsilon$.
Applying the same consistency condition to the remaining $\epsilon$ and $\eta$-overlaps further implies $\sigma=\tau=\upsilon$.
Hence, all $4$ permutations must coincide:
$\pi=\sigma=\tau=\upsilon$. Thus, the apparent $(k!)^4$ combinatorial freedom collapses to a single common permutation in Eq.~\eqref{eq:3step_FP_perfect}.
Consequently, the long-time limit is Haar-compatible $F^{(k)}_{\nu}(\infty)=k!$,
rather than $(k!)^4$. This is an exact demonstration of how the additional quench in the 3SP provides enough independent constraints on the overlap matrices to suppress all non-Haar contributions. Therefore, in the long-time regime, the 3SP achieves the Haar value and thus realizes a unitary $k$-design up to finite-$T$ and finite-size corrections. 

{\em \color{blue} Analysis using random matrix Hamiltonian---} 
However, almost no physically realistic Hamiltonian family strictly satisfies the flat-overlap condition $|U_{mn}|^2 = 1/D$. This naturally leads us to consider a broader and more generic setting \cite{Altland:2021rqn}. To this end, let us consider sampling the Hamiltonians of the 3SP from a random-matrix ensemble. For simplicity, we focus on the GUE Hamiltonian first\cite{Collins2002,Mehta2004,Collins_2006,Anderson2009}. Each GUE Hamiltonian has the spectral decomposition $H_l = W_l \Lambda_l W_l^\dagger$,
where $\Lambda_l$ is diagonal and $W_l$ is Haar random. Here, $l=1,2,3$ labels the distinct Hamiltonian realizations in the 2SP or 3SP protocol. Consequently, the eigenbasis overlap matrices are themselves unitary with $U^{(\bm{1})}_{mn} \equiv \braket{E_m|\epsilon_n} = (W_1^\dagger W_2)_{mn}$,
    $U^{(\bm{2})}_{pg} \equiv \braket{\epsilon_p|\eta_g} = (W_2^\dagger W_3)_{pg}$.
Because $W_1$, $W_2$, and $W_3$ are independent Haar-random unitaries, they are invariant under left and right multiplication. Both $U^{(\bm{1})}=W_1^\dagger W_2$ and $U^{(\bm{2})}=W_2^\dagger W_3$ are Haar distributed, and in fact are mutually independent.

Consider the perfect time-filtering case, where the permutation structure is as shown in Eqs.~\eqref{eq:2step_FP_perfect} and~\eqref{eq:3step_FP_perfect}. Averaging over the Haar-random overlap matrices gives the following result for the general $F^{(k)}$ in the 2SP and 3SP. Detailed proofs using Weingarten function are provided in SM~\cite{SM}.
\begin{theorem}[Perfect-filter limits of the 2SP and 3SP]
\label{thm:perfect_filter}
In the perfect-filter limit $T\to\infty$, let
$U^{(\bm{1})},U^{(\bm{2})}\in U(D)$ be the independent Haar-random
overlap matrices. Assume that the frame potentials are self-averaging in the large-$D$ limit.
Then, as $D\to\infty$ with $k$ fixed,
\begin{align}
\mathbb E_{U^{(\bm{1})}}\!\left[F_{\mathrm{2SP}}^{(k)}\right]
&=k!\sum_{j=0}^k \binom{k}{j}\,!(k-j)\,2^j+O(D^{-1}),\\
\mathbb E_{U^{(\bm{1})},U^{(\bm{2})}}\!\left[F_{\mathrm{3SP}}^{(k)}\right]
&=k!+O(D^{-1}),
\end{align}
where $!n$ is the derangement number.
\end{theorem}



The theorem reveals that compared with 2SP, one extra quench on 3SP imposes stronger constraints on the energy eigenbasis, leading to much better approximation to a unitary $k$-design.

{\em \color{blue} General symmetry classes---}
Theorem~\ref{thm:perfect_filter} was stated for the GUE. However, interestingly, the superiority of 3SP still exists for other random matrix symmetry classes. As shown in SM~\cite{SM}, we argue this result by directly analyzing the leading-order contribution instead of providing the strict orthogonal and symplectic Weingartein calculus. 

We find that the leading order contribution in 2SP depends on the details of the random matrix statistics, which makes it symmetry class depended. Starting from  Eq.~\eqref{eq:2step_FP_perfect}, we can write $
\prod_{a=1}^k |U_{m_a n_a}|^2\,|U_{m_{\pi(a)} n_{\sigma(a)}}|^2
=\prod_{r=1}^{k}|U_{m_r n_r}|^2\,\bigl|U_{m_r\, n_{\sigma\pi^{-1}(r)}}\bigr|^2$. For each choice of $\sigma\pi^{-1}$, it can lead to evaluations of diagonal term $ \mathbb E |U_{m_r n_r}|^4 \equiv \mu_{\beta}M_{\beta}$ or off-diagonal term $\mathbb E |U_{m_r n_r}|^2 |U_{m_r n_r'}|^2\equiv M_{\beta} $ in different replica pairs. Here $\beta=1,2,4$ is the Dyson index that labels GOE,GUE,GSE seperately. Finally, consider all types of permutation the results gives
\begin{equation}
\label{eq:2SP_class_main}
F^{(k)}_{\mathrm{2SP}}(\infty)
= \lambda_{\beta}^{\,k}\; k!\sum_{j=0}^{k}\binom{k}{j}\,!(k-j)\,\mu_{\beta}^{\,j}
\bigl[1+O(k^{2}/D)\bigr],
\end{equation}
with $(\mu_\beta,\lambda_\beta)=\bigl(3,\tfrac{D}{D+2}\bigr)$, $\bigl(2,\tfrac{D}{D+1}\bigr)$, $\bigl(\tfrac32,\tfrac{4D}{D+1}\bigr)$. The $\beta=2$ case reproduces the 2SP result of Theorem~\ref{thm:perfect_filter}.

The 3SP, by contrast, is independent of the symmetry class at leading order, since class-dependent effects arise only from subleading diagonal contributions. First, as before, averaging products of overlap-matrix elements rather than their squared moduli, such as $\mathbb{E}\!\left[U^{(1)*}_{m_a p_a}U^{(1)}_{m_a p_{\sigma\pi^{-1}(a)}}\right]$, enforces all permutations $\pi$,$\sigma$,$\tau$,$\upsilon$ equal in Eq.~\eqref{eq:3step_FP_perfect}. Consequently, the off-diagonal terms, $\mathbb{E}\!\left[|U_{m_r p_r}|^2|U_{m_r f_r}|^2\right]$ with $p_r\neq f_r$, provide the leading contribution, while the class-dependent diagonal terms are subleading. A detailed evaluation yields the same leading contribution for all three symmetry classes
$F^{(k)}_{\mathrm{3SP}}(\infty)=k!+O(1/D)$ for $\beta=1,2,4$.

{\em \color{blue} Imperfect time-filter error---} The perfect-filter analysis assumes $T\to\infty$, so that the time filter $I_T(\Delta E)$ reduces to a Kronecker delta, enforcing exact energy-index matching. At finite $T$, this matching is imperfect: the permutation constraints among the indices are relaxed compared with Eqs.~\eqref{eq:2step_FP_perfect} and \eqref{eq:3step_FP_perfect}. To quantify this leakage, we decompose $I_T(E_m-E_{m'})=\delta_{mm'}+\widetilde I_{mm'}(T)$ 
where the leakage term $\widetilde I_{mm'}(T)\equiv (1-\delta_{mm'})\,I_T(E_m-E_{m'})$.
We characterize the typical off-diagonal weight by
\begin{equation}
\varepsilon_H(T)\equiv \frac{1}{D(D-1)}\sum_{m\neq m'} I_T(E_m-E_{m'}),
\label{eq:error_eps_def}
\end{equation}
which depends on the Hamiltonian spectrum.

For the uniform time distribution with filter function given in Eq.~\eqref{eq:IT_sinc2}, the individual filter function for a typical nonzero gap $\Delta E$ scales as $I_T(\Delta E)=\operatorname{sinc}^2(\Delta E T/2)$, so that for $T\Delta E\gg 1$ one has $\varepsilon_H(T)=O\!\left((T\Delta E)^{-2}\right)$.
For a standard Wigner-scaled GUE Hamiltonian, the bulk mean level spacing scales as $\Delta E=\Theta(1/D)$, which gives $
\varepsilon_H(T)=O\!\left(D^2/T^2\right)$. We now bound the finite-$T$ corrections. The leakage counting is purely combinatorial and thus applies equally to different symmetry classes $\beta=1,2,4$, detailed proofs are in SM~\cite{SM}.
\begin{theorem}[Finite-time filter]
\label{thm:imperfect}
Let the overlap matrices be Haar-random on the group of the relevant symmetry class and self-averaging at large $D$, and let $F^{(k)}_{\nu}(T)$ denote the $k$-th frame potential evaluated with the finite-time filter $I_T(\Delta E)$. Then, with $\varepsilon_H(T)$ defined in Eq.~\eqref{eq:error_eps_def},
\begin{align}
\label{eq:thm_imperfect_2SP}
\mathbb E\!\left[\bigl|F^{(k)}_{\mathrm{2SP}}(T)-F^{(k)}_{\mathrm{2SP}}(\infty)\bigr|\right]
&=O\!\bigl(D\,\varepsilon_H(T)\bigr)+O(D^{-1}),\\
\label{eq:thm_imperfect_3SP}
\mathbb E\!\left[\bigl|F^{(k)}_{\mathrm{3SP}}(T)-F^{(k)}_{\mathrm{3SP}}(\infty)\bigr|\right]
&=O\!\bigl(\varepsilon_H(T)\bigr)+O(D^{-1}),
\end{align}
where $F^{(k)}_{\nu}(\infty)$ are the perfect-filter values in Theorem~\ref{thm:perfect_filter}, with Eq.~\eqref{eq:2SP_class_main} giving the 2SP value for general $\beta$.
\end{theorem}

The mechanism is a single counting step. Leakage replaces one diagonal factor by an off-diagonal sum over a new index $\tilde m\neq m$, which therefore contributes one extra free summation, i.e.\ a factor $D$. In the 3SP this leakage kills leading contribution, and the subleading factor $1/D$ enters and leaving $O(\varepsilon_H(T))$.

Hence the 3SP achieves a fixed accuracy with a parametrically shorter filter window, with suppression $O(1/\sqrt{D})$ compared with 2SP, proving the advantage of 3SP in practical realizations. 
Physically, the additional energy indices introduced by the extra quenched Hamiltonian also suppress the leakage of energy levels due to finite time windows. For uniform time sampling, the finite-$T$ correction scales as $O(D^2/T^2)$ for 3SP, requiring $T\gtrsim D/\sqrt{\delta}$ for accuracy $\delta$. However, this timescale is not a fundamental lower bound and may be reduced by optimizing time sampling or introducing limited Hamiltonian randomization or random Pauli rotations.

\begin{figure}
    \centering
    \includegraphics[width=\linewidth]{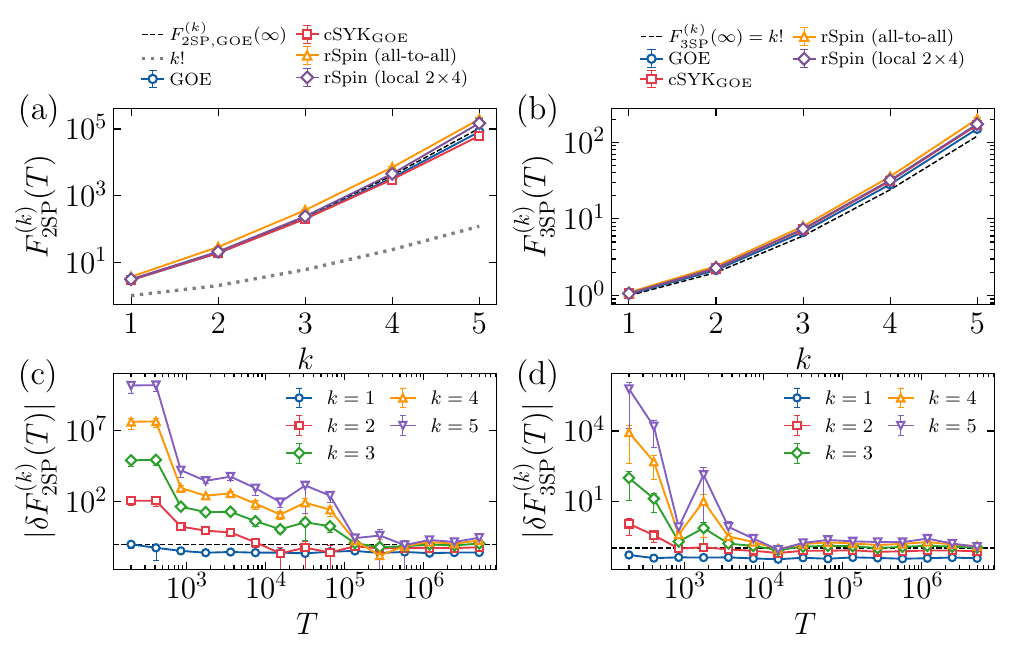}
    \caption{Numerical FPs for the 2SP and 3SP under GOE, cSYK, and random-spin dynamics in the orthogonal symmetry class. We use $D=100$ for GOE and $N=8$ at half-filling for cSYK and the two random-spin settings defined by Eq.~\eqref{eq:random-spin_Hamiltonion}. Each data point is averaged over $10^{5}$ time-sampling realizations. (a),(b) FP versus design order $k$ for 2SP and 3SP. The black dashed lines indicate the infinite-time predictions: the orthogonal-class value $F^{(k)}_{\mathrm{2SP,GOE}}(\infty)$  in Eq.~\eqref{eq:2SP_class_main} in (a), and $F^{(k)}_{\mathrm{3SP}}(\infty)=k!$ from Theorem~\ref{thm:perfect_filter} in (b). The Haar value $k!$ is shown as a gray dotted line. We use $T\approx 10^6$
    to approximate the $T\to\infty$ limit.
    (c),(d) Relative finite-time deviations for the GOE 2SP and 3SP, respectively, plotted versus the filter window $T$. The horizontal dashed line indicates the $10^{-1}$ threshold.}
    \label{fig:numerics}
\end{figure}

{\em \color{blue} Numerics---} To verify our analytical results, we numerically evaluate the FP for 2SP and 3SP using three Hamiltonian ensembles. We first use GOE random matrices as a clean benchmark, matching the symmetry class of the two experimentally motivated models considered below. We generate $H=\frac{X+X^T}{\sqrt{2D}},  \ X_{ij}\sim\mathcal{N}(0,1)$ with independent entries. As a fermionic example, we consider the complex SYK (cSYK) model~\cite{SachdevYe1993PRL,Kitaev2015Talks,MaldacenaStanford2016PRD,MaldacenaStanfordYang2016PTEP,BagretsAltlandKamenev2016NPB,GuQiStanford2017JHEP,SongJianBalents2017PRL},
\begin{equation}
    H=\sum_{i<j}\sum_{k<l}
    J_{ij;kl}\,
    c_i^\dagger c_j^\dagger c_k c_l
    +\mathrm{H.c.},
\end{equation}
with $J_{ij;kl}\sim
\mathcal{N}_{\mathbb{C}}
\!\left(0,\frac{6J^2}{N^3}\right)$. Motivated by possible realizations in cavity- and circuit-QED platforms~\cite{2023cavity,Cao_2020,baumgartner2024quantumsimulationsachdevyekitaevmodel,baumgartner2025quantumsimulationusingtrotterized}, we take $N=8$ and restrict to the half-filling sector $Q=N/2$, which belongs to the GOE symmetry class. We describe both random-spin settings by the unified Hamiltonian
\begin{equation}
    H_{\mathcal E}=\sum_{(i,j)\in\mathcal E}
    \left[J_{ij}\left(S_i^xS_j^x+S_i^yS_j^y\right)
    +V_{ij}S_i^zS_j^z\right]
    +\sum_i h_iS_i^z.
    \label{eq:random-spin_Hamiltonion}
\end{equation}
For the all-to-all random spin (rSpin) model, $\mathcal E=\{(i,j):i<j\}$, $V_{ij}=-2J_{ij}$, and $J_{ij}\sim\mathcal N(0,4J^2/N)$. This setting mimics random dipolar-spin networks relevant to nuclear magnetic resonance~\cite{Suter06,Suter07,liEmergentUniversalQuench2024,liErrorResilientReversalQuantum2026} and NV centers~\cite{Lukin17NV,Lukin18NV}. For the local rSpin model, $\mathcal E$ contains the nearest-neighbor interaction of an open $2\times4$ lattice. The couplings $J_{ij}$ and $V_{ij}$ are independent random signs with magnitude $2J/\sqrt N$. This setting represents independently tunable hopping and density interactions in superconducting circuits~\cite{nguyenProgrammableHeisenbergInteractions2024}. In both cases, $h_i\sim\mathcal N(0,h^2)$. We use $N=8$, $J=1$, $h=0.2$, and restrict to the half-filling sector.

Figure~\ref{fig:numerics}(a) shows that 2SP does not realize a unitary $k$-design for $k>1$. The GOE FP agrees with the orthogonal-class prediction $F^{(k)}_{\mathrm{2SP,GOE}}(\infty)$ in Eq.~\eqref{eq:2SP_class_main}, and remains parametrically above the Haar value $k!$. The cSYK and local U(1) data closely track the same prediction. The all-to-all rSpin FP lies above it, indicating more structured eigenbasis overlaps for that realization.
Fig.~\ref{fig:numerics}(b) shows that 3SP brings all four ensembles close to the Haar prediction $F^{(k)}_{\mathrm{3SP}}(\infty)=k!$ from Theorem~\ref{thm:perfect_filter}, including both local and all-to-all model.

We further quantify finite-$T$ effects using the relative deviation $|\delta F^{(k)}_{\nu}(T)|\equiv\bigl|F^{(k)}_{\nu}(T)/F^{(k)}_{\nu}(\infty)-1\bigr|$, where $\nu=\mathrm{2SP},\mathrm{3SP}$ and the 2SP denominator is the GOE prediction in Eq.~\eqref{eq:2SP_class_main}. The results are shown in Figs.~\ref{fig:numerics}(c) and \ref{fig:numerics}(d). In both protocols, the error decreases with increasing $T$, consistent with the scaling of Theorem~\ref{thm:imperfect}. Defining $T^*$ as the smallest $T$ such that $|\delta F^{(k)}(T)|<\gamma$ with $\gamma=10^{-1}$, we find a clear separation of time scales: for $k=3$,
$T^*_{\mathrm{2SP}}\approx 10^{5}$ but
$T^*_{\mathrm{3SP}}\approx 10^{3}$. 
This agrees with the analytical prediction that, up to nonuniversal finite-$D$ prefactors, the 3SP reaches a fixed accuracy with a parametrically shorter filter window than the 2SP.

\begin{figure}[ht]
    \centering
    \includegraphics[width=1.0\linewidth]{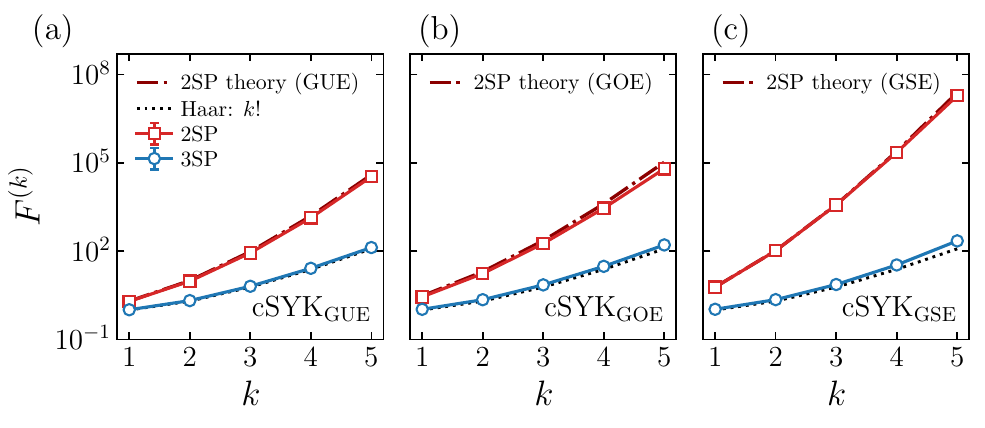}
    \caption{Symmetry-class independence of the 3SP. FPs of the 2SP and 3SP for the cSYK model in charge sectors realizing the three Wigner--Dyson classes: (a) $\mathrm{cSYK_{GUE}}$ ($N=8$, $Q=3$, $D=56$), (b) $\mathrm{cSYK_{GOE}}$ ($N=8$, $Q=4$, $D=70$), and (c) $\mathrm{cSYK_{GSE}}$ ($N=10$, $Q=5$, $D=252$). We use $T\approx 1.2\times 10^{6}$, uniform time sampling, and $10^{5}$ ($5\times 10^{4}$ for GSE) realizations for a fixed draw of $(H_1,H_2,H_3)$. The dark red dash-dotted lines show the 2SP prediction for the corresponding symmetry class, Eq.~\eqref{eq:2SP_class_main}, which reduces to Theorem~\ref{thm:perfect_filter} for GUE in panel (a). The black dotted line is the Haar value $k!$. The 2SP FP is strongly class-sensitive, while the 3SP FP converges to $k!$ in all three classes.}
    \label{fig:rmt_sectors}
\end{figure}

Finally, we test Eq.~\eqref{eq:2SP_class_main} and directly verify that the 3SP is independent of the symmetry class by exploiting the fact that different charge sectors of the cSYK model realize all three Wigner--Dyson classes. At half-filling, the model belongs to the GOE for system size $N\equiv 0 \pmod 4$ and to the GSE for $N\equiv 2 \pmod 4$. In addition, sectors away from half filling with $Q\neq N/2$ realize the GUE. As shown in Fig.~\ref{fig:rmt_sectors}, the 2SP FP depends strongly on the symmetry class and remains far from the Haar value, in agreement with Eq.~\eqref{eq:2SP_class_main}. In contrast, the 3SP FP converges to the Haar value $k!$ in all three classes, with residual deviations at $k=4,5$ consistent with the finite $D$ and finite $T$ corrections in Theorem~\ref{thm:imperfect}. The 3SP therefore provides a construction of unitary $k$ designs that is independent of the symmetry class, confirming the universality of the protocol and broadening its scope for practical applications.

{\em \color{blue} Outlook---} In this Letter, we show that unitary $k$-designs can be generated with minimal control using a quenched temporal ensemble. While the 2SP fails for higher-order designs, the 3SP realizes general unitary $k$-designs through the combined effects of temporal dephasing and phase cancellations in eigenbasis overlaps. We also quantify finite-$T$ corrections from imperfect temporal filtering. These results identify a minimal dynamical route to unitary $k$-designs using only three fixed chaotic Hamiltonians. More importantly, this route is symmetry-class independent: while the 2SP FP sharply distinguishes the three Wigner-Dyson classes, the additional quench of the 3SP erases the class information and drives the ensemble to Haar in all three classes.

The simplicity of the 3SP makes it attractive for quantum simulation, where evolution under a few fixed Hamiltonians can be substantially easier to implement than deep random circuits or Hamiltonian randomization. It is thus natural to explore the protocol experimentally on platforms already used to probe scrambling and randomized observables~\cite{Vermersch_2018,Vermersch_2019,Elben_2019,PhysRevLett.120.050406,Joshi_2020,Elben_2022,zhou2026measuringrenyientropyusing}, as well as in hybrid schemes combining temporal randomization with digital control or Hamiltonian randomization~\cite{zhou2025realizingunitarykdesignssingle}. More broadly, our results connect unitary design formation to the statistical structure of Hamiltonian eigenbases. This provides a natural link to $k$-freeness, the eigenstate thermalization hypothesis, and Hilbert-space ergodicity~\cite{PhysRevX.14.041059}. Our work suggests that higher-order ergodicity depends not only on evolution time but also on the structure of the temporal ensemble. An open question is how closely few-quench ensembles can approach Haar randomness as a function of quench number, time sampling, and chaoticity.

 {\em \color{blue} Acknowledgments---} We thank Sara Murciano, Chang Liu, Pengfei Zhang, Ning Sun, Michelle Xu, and Wen Wei Ho for helpful discussions. T.-G. Z. was supported by the Swiss National
Science Foundation under Division II (Grant No. 200020-219400).

\bibliography{ref.bib}

\clearpage
\appendix

\onecolumngrid

\section{Appendix A: Haar random and Weingarten Calculus}

For four sequences $\mathbf{i}=\left(i_1,\dots,i_p\right)$, $\mathbf{j}=\left(j_1,\dots,j_p\right)$, $\mathbf{i}^\prime=\left(i_1^\prime,\dots,i_p^\prime\right)$ and $\mathbf{j}^\prime=\left(j_1^\prime,\dots,j_p^\prime\right)$, a convenient closed form expression for averages of unitary matrices is given by

\begin{equation}\label{eq:haar_average}
    \int dU \ U_{i_1j_1}\ldots U_{i_pj_p} \left(U_{i^\prime_1j^\prime_1}\ldots U_{i^\prime_p,j^\prime_p}\right)^{\ast}   =\sum_{\sigma,\tau\in S_p}\delta_\sigma(\mathbf{i},\mathbf{i}^\prime)\delta_\tau(\mathbf{j},\mathbf{j}^\prime)\text{Wg}_D([\sigma\tau^{-1}])\, \quad 
\end{equation}

where $\sigma$ is the conjugacy class of element $\sigma$, $\text{Wg}_d([\sigma\tau^{-1}])$ is the unitary Weingarten function and invariant under the conjugacy class.  $U$ is a Haar-random $d\times d$ unitary matrix, $dU$ is the Haar measure over $U(d)$. And the index permutation delta function is defined as 
\begin{equation}
\delta_\sigma(\mathbf{i},\mathbf{i}^\prime) = \prod_{s=1}^p\delta_{i_{\sigma(s)},i_s^\prime}.
\end{equation}

If the lengths of $U$ and $U^*$ are different, then the result Eq.~\eqref{eq:haar_average} is zero.

Now we consider the large $D$ limit and fix $p$. The unitary Weingarten function satisfies
\begin{equation}
\label{eq:Wg_asymp}
\mathrm{Wg}_D([\id])=D^{-p}+O(D^{-p-2}),
\qquad
\mathrm{Wg}_D([\pi]\neq [\id])=O(D^{-p-1}),
\end{equation}
so the identity class dominates. Plugging \eqref{eq:Wg_asymp} into \eqref{eq:haar_average}
 yields the leading rule
\begin{equation}
\label{eq:leading_rule}
\int dU\ \prod_{a=1}^p U_{i_a j_a}\ \prod_{a=1}^p U^{*}_{i'_a j'_a}
=
D^{-p}\sum_{\sigma\in S_p}
\delta_\sigma(\mathbf i,\mathbf i')\,
\delta_\sigma(\mathbf j,\mathbf j')
\;+\;O(D^{-p-1}).
\end{equation}
Equation \eqref{eq:leading_rule} is the only Haar input used below. Throughout, we take $k$ fixed while $D\to\infty$.

\subsection{Appendix B: Proof of the 2SP in Theorem 1 (leading large-$D$)}
\label{sec:thm1}

\begin{proof}
    The quantity of 2SP in Theorem 1 reads
\begin{equation}
\label{eq:F1_def}
    \begin{split}
        F^{(k)}_{\nu}&= 
\sum_{\pi,\sigma \in S_k}\sum_{\{m_a\},\{n_a\}=1}^D   \left[\prod_{a=1}^k \Bigl|\langle E_{m_a}|\epsilon_{n_a}\rangle \Bigr|^2\,\Bigl|\langle E_{m_{\pi(a)}}|\epsilon_{n_{\sigma(a)}}\rangle \Bigr|^2\right] \\
 &= \sum_{\pi,\sigma \in S_k}\sum_{\{m_a\},\{n_a\}=1}^D   \left[\prod_{a=1}^k |U_{m_a n_a}|^2|U_{m_{\pi(a)} n_{\sigma(a)}}|^2\right]\\
    \end{split}
\end{equation}
with $U\in U(D)$ Haar-random.

To perform the Haar random calculus, we can write down the index $(\mathbf i,\mathbf j,\mathbf i',\mathbf j')$.
For fixed outer permutations $(\pi,\sigma)\in S_k\times S_k$, we have
\begin{equation}
\prod_{a=1}^k |U_{m_a n_a}|^2\,|U_{m_{\pi(a)} n_{\sigma(a)}}|^2
=
\prod_{r=1}^{2k} U_{i_r j_r}\ \prod_{r=1}^{2k} U^*_{i'_r j'_r},
\end{equation}
with $p=2k$ and the ordered sequences
\begin{equation}
\label{eq:thm1_indices}
\mathbf i=(m_1,\dots,m_k,\ m_{\pi(1)},\dots,m_{\pi(k)}),\qquad
\mathbf i'=\mathbf i,
\end{equation}
\begin{equation}
\label{eq:thm1_indices_j}
\mathbf j=(n_1,\dots,n_k,\ n_{\sigma(1)},\dots,n_{\sigma(k)}),\qquad
\mathbf j'=\mathbf j.
\end{equation}

Applying \eqref{eq:leading_rule} with $p=2k$ gives, for each fixed $(\pi,\sigma)$,
\begin{equation}
\label{eq:thm1_leading_integrand}
\mathbb E_U\!\left[\prod_{a=1}^k |U_{m_a n_a}|^2\,|U_{m_{\pi(a)} n_{\sigma(a)}}|^2\right]
=
D^{-2k}\sum_{\alpha\in S_{2k}}
\delta_\alpha(\mathbf i,\mathbf i)\,
\delta_\alpha(\mathbf j,\mathbf j)
\;+\;O(D^{-2k-1}).
\end{equation}

The next step is to analyze how to perform the index summation for all the $\alpha$.
In $\mathbf i$ each label $m_r$ appears \emph{exactly twice}, at positions
\begin{equation}
P^{(m)}_r=\{\,r,\ k+\pi^{-1}(r)\,\},\qquad r=1,\dots,k.
\end{equation}
If $\alpha$ maps any position out of its pair $P^{(m)}_r$, then
$\delta_\alpha(\mathbf i,\mathbf i)$ forces an identification $m_r=m_{r'}$ with $r\neq r'$,
reducing the number of free $m$-sums by at least one and hence suppressing the contribution
by at least $O(1/D)$ after the prefactor $D^{-2k}$. Thus, at leading order, $\alpha$ must preserve every
pair $P^{(m)}_r$, i.e.\ $\alpha$ belongs to the size-$2^k$ subgroup
$G_m(\pi)\cong(\mathbb Z_2)^k$ generated by independent swaps within each $P^{(m)}_r$.

Similarly, in $\mathbf j$ each $n_r$ appears twice at
\begin{equation}
P^{(n)}_r=\{\,r,\ k+\sigma^{-1}(r)\,\},\qquad r=1,\dots,k,
\end{equation}
so leading contributions require $\alpha\in G_n(\sigma)\cong(\mathbb Z_2)^k$.

Therefore,
\begin{equation}
\sum_{\{m\},\{n\}}\delta_\alpha(\mathbf i,\mathbf i)\delta_\alpha(\mathbf j,\mathbf j)
=
\begin{cases}
D^{2k} & \alpha\in G_m(\pi)\cap G_n(\sigma),\\
O(D^{2k-1}) & \text{otherwise},
\end{cases}
\end{equation}
and \eqref{eq:thm1_leading_integrand} implies
\begin{equation}
\label{eq:thm1_pair_intersection}
\sum_{\{m\},\{n\}}
\mathbb E_U\!\left[\prod_{a=1}^k |U_{m_a n_a}|^2\,|U_{m_{\pi(a)} n_{\sigma(a)}}|^2\right]
=
|G_m(\pi)\cap G_n(\sigma)|\;+\;O(1/D).
\end{equation}

To compute $|G_m(\pi)\cap G_n(\sigma)|$, the goal is to find how many indexes share same permutation result for both $\pi$ and $\sigma$. Specifically,
we define the relative permutation $\rho=\pi^{-1}\sigma\in S_k$.
The $m$-pair $P^{(m)}_r$ equals the $n$-pair $P^{(n)}_r$ iff
$\pi^{-1}(r)=\sigma^{-1}(r)$, i.e.\ $\rho(r)=r$.
For each fixed point of $\rho$ one has an independent binary choice (swap or not) that preserves both
pairings, while for $r$ not a fixed point the choice is only non-swap.
Hence
\begin{equation}
\label{eq:intersection_size}
|G_m(\pi)\cap G_n(\sigma)|=2^{\mathrm{fix}(\rho)},\qquad \rho=\pi^{-1}\sigma.
\end{equation}
For example, for $\sigma(m_1 m_2 m_3)=m_3 m_1 m_2$, $\tau(m_1 m_2 m_3)=m_2 m_1 m_3$, there is one fix point is corresponding to $m_1$ but $m_2, m_3$ is not the fix point. 

Finally we sum over outer permutations, and
combine \eqref{eq:F1_def}, \eqref{eq:thm1_pair_intersection}, and \eqref{eq:intersection_size},
\begin{equation}
\mathbb E_U F^{(k)}_\nu
=
\sum_{\pi,\sigma\in S_k} 2^{\mathrm{fix}(\pi^{-1}\sigma)} + O(1/D).
\end{equation}
For each fixed $\pi$, the map $\sigma\mapsto \rho=\pi^{-1}\sigma$ is a bijection of $S_k$, giving
\begin{equation}
\label{eq:thm1_final}
\lim_{D\to\infty}\mathbb E_U F^{(k)}_\nu
=
k!\sum_{\rho\in S_k} 2^{\mathrm{fix}(\rho)} = k!\sum_{j=0}^k \binom{k}{j}\ !(k-j)2^j,
\end{equation}
where $!(n)$ is the derangement number, where the number of permutations of 
$n$ elements with no fixed points. It has the iteration expression $!n = (n-1)(!(n-1)+!(n-2))$. The formula \eqref{eq:thm1_final} means you choose $j$ indexes from $k$ indexes, and these $j$ indexes are fixed points and therefore contribute to $2^j$, and the $(k-j)$ indexes must not be fixed point and contribute $!(k-j)$. Therefore,
this is the leading large-$D$ result stated in the main text.
\end{proof}

\subsubsection{An Example in Theorem 1}

We consider an example of $k=3$ to illustrate the idea in the proof.
We illustrate the pairing/intersection rule leading to
$|G_m(\pi)\cap G_n(\sigma)|=2^{N_{\rm match}}$ for a concrete choice of
outer permutations.  Take
\begin{equation}
\pi(m_1m_2m_3)=(m_2m_1m_3)\quad\text{with}\  \pi=(12),
\qquad
\sigma(m_1m_2m_3)=(m_3m_1m_2)\quad\text{with}\ \sigma=(132).
\end{equation}
For $k=3$ we have $p=2k=6$ and the ordered index sequences (with positions $1,\dots,6$) are
\begin{equation}
\label{eq:SM_i_j_sequences}
\mathbf i=(m_1,m_2,m_3,m_{\pi(1)},m_{\pi(2)},m_{\pi(3)})=(m_1,m_2,m_3,m_2,m_1,m_3),
\end{equation}
\begin{equation}
\mathbf j=(n_1,n_2,n_3,n_{\sigma(1)},n_{\sigma(2)},n_{\sigma(3)})=(n_1,n_2,n_3,n_3,n_1,n_2).
\end{equation}

Now we consider the pair structure. Each label appears exactly twice. The $m$-pairs (fixed by $\pi$) are $P^{(m)}_r=\{r,\ k+\pi^{-1}(r)\}$ with $r=1,2,3$,
which gives
\begin{equation}
P^{(m)}_1=\{1,5\},\qquad P^{(m)}_2=\{2,4\},\qquad P^{(m)}_3=\{3,6\}.
\end{equation}
Similarly, the $n$-pairs (fixed by $\sigma$) are
\begin{equation}
P^{(n)}_r=\{\,r,\ k+\sigma^{-1}(r)\,\},\qquad r=1,2,3,
\end{equation}
which gives
\begin{equation}
P^{(n)}_1=\{1,5\},\qquad P^{(n)}_2=\{2,6\},\qquad P^{(n)}_3=\{3,4\}.
\end{equation}

At leading order in large $D$, the Weingarten-leading term enforces that the
permutation $\alpha\in S_{2k}=S_6$ must preserve simultaneously the $m$-pairing
and the $n$-pairing.  In this example, the \emph{only} common pair is
\begin{equation}
P^{(m)}_1=P^{(n)}_1=\{1,5\},
\end{equation}
while the remaining pairs conflict. Hence
\begin{equation}
G_m(\pi)\cap G_n(\sigma)=\{\id,\,(1\,5)\}
\end{equation}
and $
\bigl|G_m(\pi)\cap G_n(\sigma)\bigr|=2$. So the only leading-order $\alpha$ are the identity and the swap within the shared pair,
$\alpha\in\{(),(1\,5)\}$ in this case of $\pi$ and $\sigma$.

\begin{figure}[t]
\centering
\begin{tikzpicture}[baseline=(M-1-1.base)]
\matrix (M) [matrix of math nodes, row sep=2.0ex, column sep=2.4ex] {
m_1 & m_2 & m_3 & \vline & m_2 & m_1 & m_3 \\
n_1 & n_2 & n_3 & \vline & n_3 & n_1 & n_2 \\
};


\draw (M-1-1.north) to[bend left=55] (M-1-6.north);
\draw (M-1-2.north) to[bend left=55] (M-1-5.north);
\draw (M-1-3.north) to[bend left=55] (M-1-7.north);

\draw (M-2-1.south) to[bend right=55] (M-2-6.south);
\draw (M-2-2.south) to[bend right=55] (M-2-7.south);
\draw (M-2-3.south) to[bend right=55] (M-2-5.south);

\draw[very thick] (M-1-1.north) to[bend left=55] (M-1-6.north);
\draw[very thick] (M-2-1.south) to[bend right=55] (M-2-6.south);
\end{tikzpicture}
\caption{Pairing lines for the example \eqref{eq:SM_i_j_sequences}.
Top row: $m$-pairs fixed by $\pi$; bottom row: $n$-pairs fixed by $\sigma$.
The only shared pair is $\{1,5\}$ (thick), giving two allowed $\alpha\in\{(),(1\,5)\}$
and thus $\bigl|G_m(\pi)\cap G_n(\sigma)\bigr|=2$.}
\label{fig:SM_pairing_k3}
\end{figure}

\subsection{Proof of 3SP in Theorem 1 (leading large-$D$)}
\label{sec:thm2}

\begin{proof}
    We now consider two independent Haar unitaries $U^{(1)},U^{(2)}\in U(D)$ and
\begin{equation}
\label{eq:F2_def}
\begin{split}
    F^{(k)}_{\nu}
&= \sum_{\pi,\sigma,\tau,\upsilon \in S_k}\sum_{\{\text{index}\}=1}^D   \prod_{a=1}^k
\Bigl(
\langle E_{m_a}|\epsilon_{p_a}\rangle\,
\langle \epsilon_{p_a}|\eta_{g_a}\rangle\,
\langle \eta_{g_a}|\epsilon_{f_a}\rangle\,
\langle\epsilon_{f_a}|E_{m_a}\rangle
\Bigr)\Bigl(
\langle E_{m_{\pi(a)}}|\epsilon_{p_{\sigma(a)}}\rangle\,
\langle \epsilon_{p_{\sigma(a)}}|\eta_{g_{\tau(a)}}\rangle\,
\langle \eta_{g_{\tau(a)}}|\epsilon_{f_{\upsilon(a)}}\rangle\,
\langle\epsilon_{f_{\upsilon(a)}}|E_{m_{\pi(a)}}\rangle
\Bigr)^{\!*} \\\\
&=
\sum_{\pi,\sigma,\tau,\upsilon\in S_k}\ \sum_{\{m,p,f,g\}=1}^D
\prod_{a=1}^k
\Bigl(U^{(1)}_{m_a p_a}\,U^{(2)}_{p_a g_a}\,U^{*(2)}_{f_a g_a}\,U^{*(1)}_{m_a f_a}\Bigr)
\Bigl(U^{(1)}_{m_{\pi(a)} p_{\sigma(a)}}\,U^{(2)}_{p_{\sigma(a)} g_{\tau(a)}}\,U^{*(2)}_{f_{\upsilon(a)} g_{\tau(a)}}\,U^{*(1)}_{m_{\pi(a)} f_{\upsilon(a)}}\Bigr)^{*},
\end{split}
\end{equation}

Because $U^{(1)}$ and $U^{(2)}$ are independent Haar, the average factorizes:
\begin{equation}
\mathbb E[F^{(k)}_\nu]
=
\sum_{\pi,\sigma,\tau,\upsilon}\ \sum_{\{m,p,f,g\}}
\mathbb E_{U^{(1)}}[\dots]\ \mathbb E_{U^{(2)}}[\dots].
\end{equation}

\paragraph{Averaging $U^{(1)}$ at leading order.}

Collect the $U^{(1)}$ factors (there are $2k$ unstarred and $2k$ starred). With $p=2k$, choose
\begin{equation}
\label{eq:U1_indices}
\mathbf i^{(1)}=(m_1,\dots,m_k,\ m_{\pi(1)},\dots,m_{\pi(k)}),\qquad
\mathbf i^{\prime(1)}=\mathbf i^{(1)},
\end{equation}
\begin{equation}
\label{eq:U1_indices_j}
\mathbf j^{(1)}=(p_1,\dots,p_k,\ f_{\upsilon(1)},\dots,f_{\upsilon(k)}),\qquad
\mathbf j^{\prime(1)}=(f_1,\dots,f_k,\ p_{\sigma(1)},\dots,p_{\sigma(k)}).
\end{equation}
Applying \eqref{eq:leading_rule} gives
\begin{equation}
\label{eq:U1_leading}
\mathbb E_{U^{(1)}}[\dots]
=
D^{-2k}\sum_{\alpha\in S_{2k}}
\delta_\alpha(\mathbf i^{(1)},\mathbf i^{(1)})\,
\delta_\alpha(\mathbf j^{(1)},\mathbf j^{\prime(1)})
+O(D^{-2k-1}).
\end{equation}

As in Theorem~1, avoiding $m$-identifications forces $\alpha$ to preserve the $m$-pairs
$P^{(m)}_r=\{r,k+\pi^{-1}(r)\}$. Among these pair-preserving $\alpha$'s, any unswapped pair produces
constraints of the form $p=\!f$ from $\delta_\alpha(\mathbf j^{(1)},\mathbf j^{\prime(1)})$, which lowers the
number of free $(p,f)$-sums and is subleading. Hence the only leading contribution is the \emph{full swap}
$\alpha=\alpha_{\rm sw}$ that swaps every $m$-pair corresponding to $P_r^{(m)}$. Evaluating $\delta_{\alpha_{\rm sw}}(\mathbf j^{(1)},\mathbf j^{\prime(1)})$
then imposes
\begin{equation}
\label{eq:U1_constraints}
\sigma=\pi,\qquad \upsilon=\pi,
\end{equation}
at leading order; all other $(\sigma,\upsilon)$ are suppressed by $O(1/D)$.

\paragraph{Averaging $U^{(2)}$ at leading order}
With the constraints \eqref{eq:U1_constraints}, collect the $U^{(2)}$ factors and choose
\begin{equation}
\label{eq:U2_indices}
\mathbf i^{(2)}=(p_1,\dots,p_k,\ f_{\pi(1)},\dots,f_{\pi(k)}),\qquad
\mathbf i^{\prime(2)}=(f_1,\dots,f_k,\ p_{\pi(1)},\dots,p_{\pi(k)}),
\end{equation}
\begin{equation}
\label{eq:U2_indices_j}
\mathbf j^{(2)}=(g_1,\dots,g_k,\ g_{\tau(1)},\dots,g_{\tau(k)}),\qquad
\mathbf j^{\prime(2)}=\mathbf j^{(2)}.
\end{equation}
Applying \eqref{eq:leading_rule} again gives
\begin{equation}
\label{eq:U2_leading}
\mathbb E_{U^{(2)}}[\dots]
=
D^{-2k}\sum_{\beta\in S_{2k}}
\delta_\beta(\mathbf i^{(2)},\mathbf i^{\prime(2)})\,
\delta_\beta(\mathbf j^{(2)},\mathbf j^{(2)})
+O(D^{-2k-1}).
\end{equation}

Now $\delta_\beta(\mathbf j^{(2)},\mathbf j^{(2)})$ enforces pair-preservation with respect to the $g$-pairs
$P^{(g)}_r=\{r,k+\tau^{-1}(r)\}$. As above, leading order requires the full swap on these pairs, and
$\delta_{\beta_{\rm sw}}(\mathbf i^{(2)},\mathbf i^{\prime(2)})$ then fixes
\begin{equation}
\label{eq:U2_constraints}
\tau=\pi.
\end{equation}

Based on the integration of $U^{(\bm{1})}$ and $U^{(\bm{2})}$, the quations \eqref{eq:U1_constraints} and \eqref{eq:U2_constraints} show that, at leading order, only the diagonal choice
\begin{equation}
\label{eq:diag_choice}
\sigma=\tau=\upsilon=\pi
\end{equation}
contributes. For such terms, the remaining index sums provide $D^{4k}$ free choices of $(m,p,f,g)$, while the two leading
Weingarten factors contribute $D^{-2k}\cdot D^{-2k}=D^{-4k}$, giving an $O(1)$ contribution equal to $1$ per $\pi$.
Therefore
\begin{equation}
\label{eq:thm2_final}
    \lim_{D\to\infty}\mathbb E_{U^{(1)},U^{(2)}} F^{(k)}_\nu
    =\sum_{\pi\in S_k} 1
    =k!.
\end{equation}
All other outer-permutation choices are suppressed by at least one power of $1/D$.
\end{proof}

\subsubsection{An example of 3SP in Theorem 1}
We illustrate the leading large-$D$ mechanism behind the diagonal constraint $\sigma=\tau=\upsilon=\pi$ for a concrete $k=3$ choice. Let $\pi(123)=(213)$ and $\pi=(12)$.

For $k=3$ we have $p=2k=6$. The $U^{(1)}$ indices used in
Eqs.~(\ref{eq:U1_indices})--(\ref{eq:U1_indices_j}) become
\begin{equation}
\mathbf i^{(1)}=(m_1,m_2,m_3,m_{\pi(1)},m_{\pi(2)},m_{\pi(3)})
=(m_1,m_2,m_3,m_2,m_1,m_3),
\end{equation}
\begin{equation}
    \mathbf j^{(1)}=(p_1,p_2,p_3, f_{\upsilon(1)},f_{\upsilon(2)},f_{\upsilon(3)}),\qquad
    \mathbf j^{\prime(1)}=(f_1,f_2,f_3, p_{\sigma(1)},p_{\sigma(2)},p_{\sigma(3)}).
\end{equation}

The sequence $\mathbf i^{(1)}$ contains each $m_r$ twice, at the pair positions
\begin{equation}
P^{(m)}_1=\{1,5\},\qquad P^{(m)}_2=\{2,4\},\qquad P^{(m)}_3=\{3,6\}.
\end{equation}
Thus $\delta_\alpha(\mathbf i^{(1)},\mathbf i^{(1)})$ is leading only if
$\alpha$ preserves these pairs. Among such pair-preserving $\alpha$'s, any \emph{unswapped}
pair forces an identification $p=\!f$ in $\delta_\alpha(\mathbf j^{(1)},\mathbf j^{\prime(1)})$.
For instance, if $\alpha$ fixes position $1$ instead of swapping $1\leftrightarrow 5$, then
$\delta_{j^{(1)}_{\alpha(1)},j^{\prime(1)}_{1}}=\delta_{p_1,f_1}$ reduces the number of free
$(p,f)$ sums by one and hence is suppressed by $O(1/D)$ after the $D^{-6}$ prefactor.
Therefore the only leading contribution is the \emph{full swap}
\begin{equation}
\alpha_{\rm sw}=(1\,5)(2\,4)(3\,6).
\end{equation}

We can directly check the delta function constraint, with $\delta_\alpha(\mathbf j,\mathbf j')=\prod_{s=1}^{6}\delta_{j_{\alpha(s)},j'_s}$. For the First half ($s=1,2,3$): since $\alpha_{\rm sw}(1)=5$, $\alpha_{\rm sw}(2)=4$, $\alpha_{\rm sw}(3)=6$,
\begin{equation}
\delta_{j_5,f_1}=\delta_{f_{\upsilon(2)},f_1},\quad
\delta_{j_4,f_2}=\delta_{f_{\upsilon(1)},f_2},\quad
\delta_{j_6,f_3}=\delta_{f_{\upsilon(3)},f_3},
\end{equation}
which at leading order (no $f$-identifications) implies $\upsilon(2)=1,\ \upsilon(1)=2,\ \upsilon(3)=3$, i.e.
$\upsilon=\pi$.

For the second halpf, we write $s=3+q$ with $q=1,2,3$, one has
$\alpha_{\rm sw}(3+q)=\pi(q)$ for the full swap on $P_r^{(m)}$, hence
\begin{equation}
\delta_{j_{\pi(1)},p_{\sigma(1)}}=\delta_{p_2,p_{\sigma(1)}},\quad
\delta_{j_{\pi(2)},p_{\sigma(2)}}=\delta_{p_1,p_{\sigma(2)}},\quad
\delta_{j_{\pi(3)},p_{\sigma(3)}}=\delta_{p_3,p_{\sigma(3)}},
\end{equation}
which at leading order (no $p$-identifications) forces $\sigma(1)=2,\ \sigma(2)=1,\ \sigma(3)=3$, i.e. $\sigma=\pi$.

Hence $U^{(1)}$-averaging enforces $\sigma=\upsilon=\pi$ at $O(1)$.

\begin{figure}[t]
\centering
\begin{tikzpicture}[baseline=(M-1-2.base)]
\matrix (M) [matrix of math nodes, row sep=2.1ex, column sep=2.2ex] {
\mathbf i^{(1)}       = & m_1 & m_2 & m_3 & m_{\pi(1)} & m_{\pi(2)} & m_{\pi(3)} \\
\mathbf i^{\prime(1)} = & m_1 & m_2 & m_3 & m_{\pi(1)} & m_{\pi(2)} & m_{\pi(3)} \\
& {} & {} & {} & {} & {} & {} \\
\mathbf j^{(1)}       = & p_1 & p_2 & p_3 & f_{\nu(1)} & f_{\nu(2)} & f_{\nu(3)} \\
\mathbf j^{\prime(1)} = & f_1 & f_2 & f_3 & p_{\sigma(1)} & p_{\sigma(2)} & p_{\sigma(3)} \\
};

\draw (M-1-2.north) to[bend left=60] (M-1-6.north); 
\draw (M-1-3.north) to[bend left=45] (M-1-5.north); 
\draw (M-1-4.north) to[bend left=35] (M-1-7.north); 
\end{tikzpicture}
\caption{Leading $U^{(1)}$ contraction for $k=3$ and $\pi=(12)$.
The full-swap $\alpha_{\rm sw}=(1\,5)(2\,4)(3\,6)$ pairs the second-half $f$'s with the first-half $f$'s
and the first-half $p$'s with the second-half $p$'s, forcing $\upsilon=\sigma=\pi$ at leading order.}
\label{fig:SM_thm2_U1_k3}
\end{figure}

Similarly, averaging $U^{(2)}$ also forces $\tau=\pi$
Therefore for $k=3$ and $\pi=(12)$, the leading term therefore requires the diagonal choice
$\sigma=\tau=\upsilon=\pi$, and the net contribution of this $\pi$ is $1$. Summing over all $\pi\in S_3$ yields
$\lim_{D\to\infty}\mathbb E[F^{(3)}_\nu]=|S_3|=3!=6$, in agreement with Eq.~\eqref{eq:thm2_final}.

\section{Appendix B: Finite-$T$ corrections from imperfect time filtering}
\label{app:finiteT_filter_leakage}

The long-time arguments in the main text rely on the fact that the time average suppresses
off-diagonal energy differences.  At finite $T$ this suppression is imperfect, and one should
track the resulting ``leakage'' quantitatively.

Throughout we use
\begin{equation}
I_T(\Delta E)\equiv
\int_0^T\!dt\int_0^T\!dt'\,P(t)P(t')\,e^{i\Delta E (t-t')}
=
\left|\int_0^T\!dt\,P(t)\,e^{i\Delta E t}\right|^2,
\label{eq:IT_def_app}
\end{equation}
which obeys the basic properties
\begin{equation}
0\le I_T(\Delta E)\le 1,\qquad I_T(0)=1,
\qquad I_T(\Delta E)\ \text{has width}\ \Delta E\sim 1/T.
\label{eq:IT_basic_props}
\end{equation}
For the normalized uniform distribution $P(t)=\frac{1}{T}\mathbf{1}_{[0,T]}(t)$,
\begin{equation}
I_T(\Delta E)
=\left(\frac{\sin(\Delta E T/2)}{\Delta E T/2}\right)^2.
\label{eq:IT_sinc2_app}
\end{equation}
In a discrete, nondegenerate spectrum, the limit $T\to\infty$ suppresses $m\neq m'$ contributions
because typical $\Delta E_{mm'}\neq 0$ implies $I_T(\Delta E_{mm'})\to 0$.
Crucially, however, at any finite $T$ there is a positive leakage $I_T(\Delta E_{mm'})>0$ for $m\neq m'$,
and the size of the correction is a \emph{weighted sum over \underline{all} off-diagonal pairs}---not only
nearest neighbors in $m$.

\subsection{2SP: structure of the leakage}
\label{app:finiteT_two_step}

Starting from Eq.(4) in the main text,
\begin{equation}
F^{(1)}_{\nu}(T)
=
\sum_{m,n,m',n'=1}^D
|U_{mn}|^2\,|U_{m'n'}|^2\;
I_T(E_m-E_{m'})\;
I_T(\epsilon_n-\epsilon_{n'}),
\qquad U_{mn}\equiv \langle E_m|\epsilon_n\rangle.
\label{eq:F1_general_app}
\end{equation}
It is convenient to separate diagonal and off-diagonal parts of the filter:
\begin{equation}
I_T(E_m-E_{m'})=\delta_{mm'}+\widetilde I^{(1)}_{mm'}(T),
\qquad
\widetilde I^{(1)}_{mm'}(T)\equiv (1-\delta_{mm'})\,I_T(E_m-E_{m'}),
\label{eq:IT_split_H1}
\end{equation}
and similarly
\begin{equation}
I_T(\epsilon_n-\epsilon_{n'})=\delta_{nn'}+\widetilde I^{(2)}_{nn'}(T),
\qquad
\widetilde I^{(2)}_{nn'}(T)\equiv (1-\delta_{nn'})\,I_T(\epsilon_n-\epsilon_{n'}).
\label{eq:IT_split_H2}
\end{equation}
Plugging \eqref{eq:IT_split_H1}--\eqref{eq:IT_split_H2} into \eqref{eq:F1_general_app} yields the exact decomposition
\begin{equation}
F^{(1)}_{\nu}(T)=F^{(1)}_{\nu}(\infty)+\Delta F^{(1)}_{(1)}(T)+\Delta F^{(1)}_{(2)}(T)+\Delta F^{(1)}_{(12)}(T),
\label{eq:F1_decomp_app}
\end{equation}
where the long-time limit is the IPR
\begin{equation}
F^{(1)}_{\nu}(\infty)=\sum_{m,n=1}^D |U_{mn}|^4,
\label{eq:F1_IPR_app}
\end{equation}
and the three leakage pieces are
\begin{align}
\Delta F^{(1)}_{(1)}(T)
&=
\sum_{\substack{m\neq m'\\ n,n'}}
|U_{mn}|^2\,|U_{m'n'}|^2\;
\widetilde I^{(1)}_{mm'}(T)\;
\delta_{nn'},
\label{eq:DeltaF1_H1_only_app}\\
\Delta F^{(1)}_{(2)}(T)
&=
\sum_{\substack{n\neq n'\\ m,m'}}
|U_{mn}|^2\,|U_{m'n'}|^2\;
\delta_{mm'}\;
\widetilde I^{(2)}_{nn'}(T),
\label{eq:DeltaF1_H2_only_app}\\
\Delta F^{(1)}_{(12)}(T)
&=
\sum_{\substack{m\neq m'\\ n\neq n'}}
|U_{mn}|^2\,|U_{m'n'}|^2\;
\widetilde I^{(1)}_{mm'}(T)\;
\widetilde I^{(2)}_{nn'}(T).
\label{eq:DeltaF1_both_app}
\end{align}
Because all terms are nonnegative, \emph{any} finite leakage increases $F^{(1)}_{\nu}$ above its long-time value.
This is the key qualitative point: for the 2SP the finite-$T$ correction is positive and appears already
at first order in the off-diagonal weight.

A compact way to parameterize the leakage is via the averaged off-diagonal weights
\begin{equation}
\varepsilon_{H}(T)\equiv \frac{1}{N_H(N_H-1)}\sum_{a\neq b} I_T(E_a-E_b),
\label{eq:eps_def_app}
\end{equation}
where $a,b$ label distinct spectral blocks and $N_H$ is their number.  Thus $N_H=D$ for a nondegenerate GOE or GUE spectrum.  For the GSE, $a,b$ label Kramers doublets and $N_H=D/2$: the two states inside a doublet belong to the perfect-filter projector and their zero energy difference is not counted as leakage.  When two Hamiltonians occur below, their leakage parameters are denoted by $\varepsilon_{H_1}$ and $\varepsilon_{H_2}$.
These numbers are $\varepsilon_{H_j}(T)\to 0$ as $T\to\infty$, but at finite $T$ they can be substantial unless
$T$ is parametrically larger than the inverse many-body level spacing (Heisenberg time).

For a nondegenerate spectrum, if one models the change-of-basis as perfectly flat,
\begin{equation}
|U_{mn}|^2=\frac{1}{D}\qquad (\text{idealized random-eigenbasis limit}),
\label{eq:flatU_app}
\end{equation}
then \eqref{eq:F1_general_app} factorizes and one gets an explicit closed form:
\begin{equation}
F^{(1)}_{\nu}(T)
=
\frac{K_{H_1}(T)\,K_{H_2}(T)}{D^2},
\qquad
K_{H}(T)\equiv \sum_{a,b=1}^D I_T(E_a-E_b).
\label{eq:F1_flatU_factor_app}
\end{equation}
Using $I_T(0)=1$ and \eqref{eq:eps_def_app},
\begin{equation}
K_{H}(T)=D + D(D-1)\,\varepsilon_{H}(T),
\label{eq:K_eps_app}
\end{equation}
so
\begin{equation}
F^{(1)}_{\nu}(T)
=
\Bigl[1+(D-1)\varepsilon_{H}(T)\Bigr]\,
\Bigl[1+(D-1)\varepsilon_{H_2}(T)\Bigr].
\label{eq:F1_flatU_eps_app}
\end{equation}
In particular, relative to the long-time value $F^{(1)}_{\nu}(\infty)=1$ in the flat-overlap model,
\begin{equation}
F^{(1)}_{\nu}(T)-1
=(D-1)\bigl[\varepsilon_{H}(T)+\varepsilon_{H_2}(T)\bigr]
+(D-1)^2\varepsilon_{H}(T)\varepsilon_{H_2}(T).
\label{eq:F1_error_scaling_app}
\end{equation}
This makes the main correction mechanism transparent: the error is controlled by the \emph{aggregate} off-diagonal leakage
$\varepsilon_H(T)$, not by a single ``nearest-neighbor'' spacing.
For the uniform $P(t)$ in \eqref{eq:IT_sinc2_app}, one has
$I_T(\Delta E)\lesssim (\Delta E T)^{-2}$ for $|\Delta E|T\gg 1$.  If the relevant nonzero block gaps are uniformly of order $\Delta E$, this gives $\varepsilon_H(T)=O((T\Delta E)^{-2})$.  In general, however, $\varepsilon_H(T)$ is the spectral average in Eq.~\eqref{eq:eps_def_app}, rather than a function of one representative or nearest-neighbor gap.

\subsection{3SP: suppression of a single leakage channel}
\label{app:finiteT_three_step}

For the 3SP, the long-time ($T\to\infty$) Haar-compatible result comes from the fact that the only
surviving terms are those in which \emph{every complex overlap amplitude is paired with its complex conjugate}
in a way consistent with the cyclic ordering around
\begin{equation}
\langle E_m|\epsilon_p\rangle\,
\langle \epsilon_p|\eta_g\rangle\,
\langle \eta_g|\epsilon_f\rangle\,
\langle \epsilon_f|E_m\rangle.
\label{eq:cycle_amp_app}
\end{equation}
At finite $T$ the time filters broaden the energy constraints and, in principle, allow additional index patterns.
The important distinction from the 2SP case is that the 3SP summand is \emph{not} a function of
$|U|^2$ only: it carries \emph{complex phases}. As a result, leakage contributions that do not implement a
full conjugate pairing acquire random phases and cancel.

A clean way to formalize this is to define the single-cycle amplitude
\begin{equation}
A_{m p g f}\equiv
\langle E_m|\epsilon_p\rangle\,
\langle \epsilon_p|\eta_g\rangle\,
\langle \eta_g|\epsilon_f\rangle\,
\langle \epsilon_f|E_m\rangle,
\label{eq:A_def_app}
\end{equation}
so that the $k$th FP is a sum over products of $A$'s and $A^*$'s multiplied by time-filter factors enforcing
(additive) energy matching conditions.

\paragraph{Random-eigenbasis (Haar) estimate.}
In chaotic systems, it is standard to model the relative eigenbases of independent Hamiltonians as Haar-random.
In that model, any term in the FP expansion that contains an overlap amplitude without a matching conjugate
has vanishing ensemble average:
\begin{equation}
\mathbb{E}_{\text{Haar}}\!\bigl[\,(\dots U_{ab}\dots)\,(\dots U_{a'b'}\dots)^{*}\bigr]=0
\quad
\text{unless the indices contract to pair each $U$ with a $U^*$}.
\label{eq:Haar_pairing_rule_app}
\end{equation}
A single leakage relaxes one additive energy constraint and breaks the full-swap contraction that gives the perfect-filter result.  It does not, however, vanish identically: the leading nonzero contraction uses a non-swap on the leaked pair.  Relative to the 2SP, this contraction also identifies one pair of intermediate indices and thereby removes the extra free sum produced by the leakage.  The detailed counting below consequently gives an $O(\varepsilon_H)$ correction for the 3SP, rather than the $O(D\varepsilon_H)$ correction of the 2SP.

\section{Appendix C: Proof of the finite-time filter theorem (Theorem 2)}
\label{app:finiteT_theorem}

\subsection{Two-step protocol}

\begin{proof}
For simplicity, we only evaluate the contribution $\Delta F^{(k)}_{(1)}$.
\begin{equation}
\label{eq:DeltaF1_def}
\begin{split}
\Delta F^{(k)}_{(1)}
&= \sum_{\pi,\sigma \in S_k}\sum_{b=1}^k \sum_{\{m_a\},\{n_a\}=1}^D  \sum_{\tilde{m}\neq m_{\pi(b)}} \left[\prod_{a=1}^k |U_{m_a n_a}|^2\right]
\left[\prod_{a\neq b} |U_{m_{\pi(a)} n_{\sigma(a)}}|^2\right]
|U_{\tilde{m}\, n_{\sigma(b)}}|^2\,
I_T(E_{m_{\pi(b)}} - E_{\tilde{m}}).
\end{split}
\end{equation}
Here $U\in U(D)$ is Haar-random. We introduce the ``partial permutation'' $\pi_b$ by
\begin{equation}
\label{eq:partial_permutation}
\pi_b(m_1 m_2 \dots m_k)
=
m_{\pi(1)} m_{\pi(2)} \dots m_{\pi(b-1)} \,\tilde{m}\, m_{\pi(b+1)} \dots m_{\pi(k)}.
\end{equation}
Physically, this encodes imperfect energy filtering: the $b$-th replica pairing on the $m$-indices is broken by a leakage $m_{\pi(b)}\to \tilde m$ (for the leading contribution one typically has $|\tilde m-m_{\pi(b)}|\approx 1$).

As before, we package indices as $(\mathbf i,\mathbf j,\mathbf i',\mathbf j')$.
For fixed $(\pi,\sigma)\in S_k\times S_k$, we can write
\[
\prod_{a=1}^k |U_{m_a n_a}|^2\,|U_{m_{\pi(a)} n_{\sigma(a)}}|^2
=
\prod_{r=1}^{2k} U_{i_r j_r}\ \prod_{r=1}^{2k} U^*_{i'_r j'_r},
\]
with $p=2k$ and the ordered sequences
\begin{equation}
\label{eq:thm3_indices_i}
\mathbf i=(m_1,\dots,m_k,\ m_{\pi(1)},\dots,m_{\pi(b-1)},\tilde{m},m_{\pi(b+1)},\dots,m_{\pi(k)}),\qquad
\mathbf i'=\mathbf i,
\end{equation}
\begin{equation}
\label{eq:thm3_indices_j}
\mathbf j=(n_1,\dots,n_k,\ n_{\sigma(1)},\dots,n_{\sigma(k)}),\qquad
\mathbf j'=\mathbf j.
\end{equation}

Using the Weingarten expansion, we obtain
\begin{equation}
\Delta F_{(1)}^{(k)} = \sum_{\alpha\in S_{2k}}\mathrm{Wg}_D(\alpha)\ \Sigma_m(\alpha)\ \Sigma_n(\alpha),
\end{equation}
where
\begin{equation}
\Sigma_m(\alpha)=\sum_{b=1}^k\sum_{\{m\},\tilde m\neq m_{\pi(b)}} I_T(E_{m_{\pi(b)}}-E_{\tilde m})\ \delta_{\alpha}(\mathbf{i},\mathbf{i}),
\qquad
\Sigma_n(\alpha)=\sum_{\{n\}} \delta_{\alpha}(\mathbf{j},\mathbf{j}).
\end{equation}

For the $m$-pairs $P^{(m)}_r=\{\,\pi(r),\,k+r\,\}$ with $r\in\{1,\dots,k\}\setminus\{b\}$, the permutation $\alpha$ may implement either swap or non-swap. In contrast, the special pair $P^{(m)}_{b}=\{\pi(b),\,k+b\}$ corresponds to $(m_{\pi(b)},\tilde m)$ and only supports the non-swap, due to $\tilde m\neq m_{\pi(b)}$. Hence $G_m(\pi_b)\cong(\mathbb{Z}_2)^{k-1}$, and
\begin{equation}
\Sigma_m(\alpha)=
\begin{cases}
D^{k-1}\displaystyle\sum_{b=1}^k \sum_{\tilde{m}\neq m_{\pi(b)}} I_T(E_{\tilde{m}}-E_{m_{\pi(b)}})
= k\,D^{k}(D-1)\,\varepsilon_{H}(T), & \alpha \in G_m(\pi_b),\\[6pt]
O\!\left(D^k\,\varepsilon_{H}(T)\right), & \text{otherwise.}
\end{cases}
\end{equation}
For the $n$-pairs, the same reasoning as in the previous derivation applies: both swap and non-swap are allowed, and $G_n(\sigma)\cong(\mathbb{Z}_2)^k$. Therefore
\begin{equation}
\Sigma_n(\alpha)=
\begin{cases}
D^{k}, & \alpha \in G_n(\sigma),\\
O(D^{k-1}), & \text{otherwise.}
\end{cases}
\end{equation}

Collecting the leading terms gives
\begin{equation}
\label{eq:thm3_pair_intersection}
\Delta F_{(1)}^{(k)}
=
k (D-1)\,\varepsilon_{H}(T)\,
\bigl|G_m(\pi_b)\cap G_n(\sigma)\bigr|
\;+\;O(1/D).
\end{equation}

If, in addition, the relevant gaps are uniformly of order $\Delta E$, the sinc filter gives
$\varepsilon_{H}(T)=O((T\Delta E)^{-2})$.
With the GUE convention $\Delta E\sim 1/D$, this becomes
\begin{equation}
\Delta F_{(1)}^{(k)} = O(D^3 T^{-2}).
\end{equation}

By a similar analysis, the other leakage contribution
\begin{equation}
\label{eq:DeltaF2_def}
\begin{split}
\Delta F^{(k)}_{(2)}
&= \sum_{\pi,\sigma \in S_k}\sum_{b=1}^k \sum_{\{m_a\},\{n_a\}=1}^D  \sum_{\tilde{n}\neq n_{\sigma(b)}}
\left[\prod_{a=1}^k |U_{m_a n_a}|^2\right]
\left[\prod_{a\neq b} |U_{m_{\pi(a)} n_{\sigma(a)}}|^2\right]
|U_{m_{\sigma(b)}\, \tilde{n}}|^2\,
I_T(E_{n_{\sigma(b)}} - E_{\tilde{n}}).
\end{split}
\end{equation}
obeys the same bound,
\begin{equation}
\Delta F_{(2)}^{(k)} = O(D^3T^{-2}).
\end{equation}

Including the other leakage positions and the subleading Haar contractions, the 2SP part of Theorem~\ref{thm:imperfect} follows:
\begin{equation}
\mathbb E\!\left[\bigl|F^{(k)}_{\mathrm{2SP}}(T)-F^{(k)}_{\mathrm{2SP}}(\infty)\bigr|\right]
=O\!\bigl(D\varepsilon_H(T)\bigr)+O(D^{-1}).
\end{equation}
\end{proof}

\subsection{Three-step protocol}

\begin{proof}
We consider two representative leakage mechanisms.

\paragraph{Leakage on $m$ (example $\Delta F_{(1)}^{(k)}$).}
\begin{equation}
\label{eq:F4_def}
\begin{split}
\Delta F^{(k)}_{(1)}
=&
\sum_{\pi,\sigma,\tau,\upsilon\in S_k}\ \sum_{\{m,p,f,g\}=1}^D \sum_{b=1}^k\sum_{\tilde{m}\neq m_{\pi(b)}}
\prod_{a=1}^k
\Bigl(U^{(1)}_{m_a p_a}\,U^{(2)}_{p_a g_a}\,U^{*(2)}_{f_a g_a}\,U^{*(1)}_{m_a f_a}\Bigr) \\
&\times
\prod_{a\neq b}\Bigl(U^{(1)}_{m_{\pi(a)} p_{\sigma(a)}}\,U^{(2)}_{p_{\sigma(a)} g_{\tau(a)}}\,U^{*(2)}_{f_{\upsilon(a)} g_{\tau(a)}}\,U^{*(1)}_{m_{\pi(a)} f_{\upsilon(a)}}\Bigr)^{*}\\
&\times
\Bigl(U^{(1)}_{\tilde{m}p_{\sigma(b)}}\,U^{*(1)}_{\tilde{m}f_{\upsilon(b)}}\Bigr)^{*}\,
I_T(E_{m_{\pi(b)}} - E_{\tilde{m}}).
\end{split}
\end{equation}

We collect the $U^{(1)}$ factors (there are $2k$ unstarred and $2k$ starred). With $p=2k$, choose
\begin{equation}
\label{eq:U1_indices_mleak}
\mathbf i^{(1)}=(m_1,\dots,m_k,\ m_{\pi(1)},\dots,m_{\pi(b-1)},\tilde{m},m_{\pi(b+1)},\dots,m_{\pi(k)}),\qquad
\mathbf i^{\prime(1)}=\mathbf i^{(1)},
\end{equation}
\begin{equation}
\label{eq:U1_indices_j_mleak}
\mathbf j^{(1)}=(p_1,\dots,p_k,\ f_{\upsilon(1)},\dots,f_{\upsilon(k)}),\qquad
\mathbf j^{\prime(1)}=(f_1,\dots,f_k,\ p_{\sigma(1)},\dots,p_{\sigma(k)}).
\end{equation}

For fixed $(\pi,\sigma)$ we write
\begin{equation}
\Delta F_{(1)}^{(k)} = \sum_{\alpha\in S_{2k}}\mathrm{Wg}_D(\alpha)\ \Sigma_m(\alpha)\ \Sigma_n(\alpha),
\end{equation}
with
\begin{equation}
\Sigma_m(\alpha)=\sum_{b=1}^k\sum_{\{m\},\tilde m\neq m_{\pi(b)}} I_T(E_{m_{\pi(b)}}-E_{\tilde m})\ \delta_{\alpha}(\mathbf{i}^{(1)},\mathbf{i}^{(1)}),
\qquad
\Sigma_n(\alpha)=\sum_{\{p,f\}} \delta_{\alpha}(\mathbf{j}^{(1)},\mathbf{j}^{\prime(1)}).
\end{equation}

As in Theorem~2, the leading contribution would come from the full swap $\alpha=\alpha_k(\pi)$, but here swapping the special pair $(m_{\pi(b)},\tilde m)$ forces $\tilde m=m_{\pi(b)}$ and hence vanishes. Thus the leading nonzero term is $\alpha=\alpha_{k-1}(\pi_b)$ (swap on $k-1$ pairs, non-swap on the special pair):
\begin{equation}
    \Sigma_m(\alpha) = 
    \begin{cases}
         D^{k-1} \sum_b \sum_{\tilde{m}\neq m_{\pi(b)}} \delta_{\tilde{m}\, m_{\pi(b)}} I_T(E_{\tilde{m}} - E_{{m_{\pi(b)}}}) = 0  & \alpha = \alpha_k(\pi)\\
        D^{k-1} \sum_b \sum_{\tilde{m}\neq m_{\pi(b)}}  I_T(E_{\tilde{m}} - E_{{m_{\pi(b)}}}) = O(D^{k+1}\varepsilon_{H}(T)) & \alpha  = \alpha_{k-1}(\pi_b)
    \end{cases}
\end{equation}

Evaluating $\delta_{\alpha}(\mathbf j^{(1)},\mathbf j^{\prime(1)})$ shows that, at leading order, $\sigma=\pi$ and $\upsilon=\pi$ are selected, and
\begin{equation}
\Sigma_n(\alpha)=
\begin{cases}
D^{k}, & \alpha=\alpha_k(\pi)\ \ \text{and}\ \ \pi=\sigma=\upsilon,\\
D^{k-1}, & \alpha=\alpha_{k-1}(\pi_b)\ \ \text{and}\ \ \pi=\sigma=\upsilon,
\end{cases}
\end{equation}
while all other $(\sigma,\upsilon)$ are suppressed by $O(1/D)$. Hence
\begin{equation}
\Delta F^{(k)}_{(1)} = O(\varepsilon_{H}(T)).
\end{equation}
When the relevant gaps are uniformly of order $\Delta E$, using $\varepsilon_{H}(T)=O((T\Delta E)^{-2})$ and $\Delta E\sim 1/D$ yields
\begin{equation}
\Delta F^{(k)}_{(1)} = O(D^2 T^{-2}).
\end{equation}

\paragraph{Leakage on $p$ (example $\Delta F_{(2)}^{(k)}$).}
For leakage on $p$ we consider
\begin{equation}
\label{eq:F4_2_def}
\begin{split}
\Delta F^{(k)}_{(2)}
=&
\sum_{\pi,\sigma,\tau,\upsilon\in S_k}\ \sum_{\{m,p,f,g\}=1}^D \sum_{b=1}^k\sum_{\tilde{p}\neq p_{\sigma^{-1}(b)}}
\prod_{a=1}^k
\Bigl(U^{(1)}_{m_a p_a}\,U^{(2)}_{p_a g_a}\,U^{*(2)}_{f_a g_a}\,U^{*(1)}_{m_a f_a}\Bigr) \\
&\times
\prod_{a\neq b}\Bigl(U^{(1)}_{m_{\pi(a)} p_{\sigma(a)}}\,U^{(2)}_{p_{\sigma(a)} g_{\tau(a)}}\,U^{*(2)}_{f_{\upsilon(a)} g_{\tau(a)}}\,U^{*(1)}_{m_{\pi(a)} f_{\upsilon(a)}}\Bigr)^{*} \\
&\times
\Bigl(U^{(1)}_{m_{\pi(b)}\tilde{p}}\, U^{(2)}_{\tilde{p}g_{\tau(b)}}\Bigr)^{*}\,
I_T(E_{p_{\sigma(b)}} - E_{\tilde{p}}).
\end{split}
\end{equation}

Collect the $U^{(1)}$ factors as before, now with
\begin{equation}
\label{eq:U1_indices_pleak}
\mathbf i^{(1)}=(m_1,\dots,m_k,\ m_{\pi(1)},\dots,m_{\pi(k)}),\qquad
\mathbf i^{\prime(1)}=\mathbf i^{(1)},
\end{equation}
\begin{equation}
\label{eq:U1_indices_j_pleak}
\mathbf j^{(1)}=(p_1,\dots,p_k,\ f_{\upsilon(1)},\dots,f_{\upsilon(k)}),\qquad
\mathbf j^{\prime(1)}=(f_1,\dots,f_k,\ p_{\sigma(1)},\dots,p_{\sigma(b-1)},\tilde{p},p_{\sigma(b+1)},\dots,p_{\sigma(k)}).
\end{equation}

For fixed $(\pi,\sigma)$ we write
\begin{equation}
\Delta F_{(2)}^{(k)} = \sum_{\alpha\in S_{2k}}\mathrm{Wg}_D(\alpha)\ \Sigma_m(\alpha)\ \Sigma_n(\alpha),
\end{equation}
where
\begin{equation}
\Sigma_m(\alpha)=\sum_{\{m\}}\delta_{\alpha}(\mathbf i^{(1)},\mathbf i^{(1)}),
\qquad
\Sigma_n(\alpha)=\sum_{b=1}^k\sum_{\{p,f\},\tilde p\neq p_{\sigma(b)}} I_T(E_{p_{\sigma(b)}}-E_{\tilde p})\ \delta_{\alpha}(\mathbf j^{(1)},\mathbf j^{\prime(1)}).
\end{equation}

Again, the full swap $\alpha=\alpha_k(\sigma)$ would force $\tilde p=p_{\sigma(b)}$ and hence vanishes; the leading nonzero contribution is $\alpha=\alpha_{k-1}(\sigma_b)$:
\begin{equation}
    \Sigma_n(\alpha) = 
    \begin{cases}
         D^{k-1} \sum_b \sum_{\tilde{p}\neq p_{\sigma(b)}} \delta_{\tilde{m}\, m_{\sigma(b)}} I_T(E_{\tilde{p}} - E_{{p_{\sigma(b)}}}) = 0  & \alpha = \alpha_k(\sigma) \ \  \pi= \sigma=\upsilon\\
        D^{k-1} \sum_b \sum_{\tilde{p}\neq p_{\sigma(b)}} \delta_{f_{\upsilon(b)}\tilde{p}} I_T(E_{\tilde{p}} - E_{{m_{\sigma(b)}}}) = O(D^{k}\varepsilon_{H}(T)) & \alpha  = \alpha_{k-1}(\sigma_b) \ \  \pi= \sigma=\upsilon
    \end{cases}
\end{equation}
while
\begin{equation}
\Sigma_m(\alpha)=
\begin{cases}
D^{k}, & \alpha=\alpha_k(\sigma),\\
D^{k}, & \alpha=\alpha_{k-1}(\sigma_b),
\end{cases}
\end{equation}
at leading order (other permutations are suppressed by $O(1/D)$). Therefore
\begin{equation}
\Delta F^{(k)}_{(2)} = O(\varepsilon_{H}(T))
= O(D^2T^{-2})
\end{equation}
under the GUE convention $\Delta E\sim 1/D$.

Combining the leading leakage channels and the subleading Haar contractions gives the 3SP part of Theorem~\ref{thm:imperfect}:
\begin{equation}
\mathbb E\!\left[\bigl|F^{(k)}_{\mathrm{3SP}}(T)-F^{(k)}_{\mathrm{3SP}}(\infty)\bigr|\right]
=O\!\bigl(\varepsilon_H(T)\bigr)+O(D^{-1}).
\end{equation}
\end{proof}

\section{Appendix D: Symmetry-class dependence of the 2SP and class independence of the 3SP}
\label{app:symmetry_classes}

This Appendix proves the result underlying Fig.3 of the main text. The
2SP frame potential depends on the Wigner--Dyson symmetry class, while the
3SP frame potential converges to the Haar value $k!$ in all three classes.
The proof uses exact moments of the overlap matrix. For $\beta=4$, we assume
that the time filter resolves distinct Kramers doublets but not the two states
within each doublet. The perfect-filter projector therefore acts on doublet
labels.

Perfect time filtering reduces the 2SP frame potential Eq.(7) to positive moments of the overlap probabilities
$|U_{mn}|^2$. These moments depend on the Dyson class
\cite{Mehta2004,Altland:2021rqn}. By contrast, the 3SP frame potential Eq.(11) contains products of overlap amplitudes around a
loop of eigenbases. Interference removes all permutation channels except the
aligned channel, whose weight is independent of
the symmetry class in the leading order. 

\subsection{Threefold-way structure of the eigenbasis overlap matrices}
\label{app:threefold_overlaps}

Let $H_1$ and $H_2$ be independent draws from the same Wigner--Dyson class
with Dyson index $\beta$, and let $U_{mn}=\braket{E_m|\epsilon_n}$ denote
their eigenbasis overlap matrix. For $\beta=2$, the GUE eigenvectors are
independent Haar-random unitaries, so $U$ is Haar distributed on $U(D)$. For $\beta=1$, the GOE eigenvectors can be chosen real, and $U$ is Haar
distributed on $O(D)$. For $\beta=4$, each GSE level is a Kramers doublet.
We label the doublets by $\bar m,\bar n\in\{1,\dots,D/2\}$. The overlap
matrix is then a $(D/2)\times(D/2)$ array of quaternionic $2\times2$ blocks $U_{\bar m\bar n}$ and is Haar distributed on $Sp(D/2)\subset U(D)$, with
\begin{equation}
    U_{\bar m\bar n}=\begin{pmatrix} a & b\\ -b^{*} & a^{*}\end{pmatrix},
    \qquad a,b\in\mathbb{C}.
\end{equation}

For $\beta=4$, the time filter resolves the doublet weight
\begin{equation}
w_{\bar m\bar n}\equiv\sum_{m\in\bar m}\sum_{n\in\bar n}|U_{mn}|^{2}
=2\bigl(|a|^{2}+|b|^{2}\bigr).
\end{equation}
We use one normalized overlap matrix squares in all three classes:
\begin{equation}
q_{mn}\equiv|U_{mn}|^2\quad(\beta=1,2),
\qquad
q_{\bar m\bar n}\equiv\frac{w_{\bar m\bar n}}2=|a|^2+|b|^2\quad(\beta=4).
\label{eq:normalized_row_weight}
\end{equation}
Thus $\sum_n q_{mn}=1$ in every class.  In scalar probability sums such as the 2SP, the weight selected by the perfect time filter is $s_\beta q_{mn}$, where $s_1=s_2=1,s_4=2$.

Let $N_1=N_2=D$ and $N_4=D/2$ denote the number of normalized squared overlap matrix in one row.  Fix a row label $m$ and introduce independent real Gaussian variables $g_{mnr}\sim\mathcal N(0,1)$, with $n=1,\ldots,N_\beta$ and $r=1,\ldots,\beta$.  Define the corresponding random variable by $x_{mn}\equiv\sum_{r=1}^{\beta}g_{mnr}^2$.  Within the fixed row, the $x_{mn}$ are independent, with $x_{mn}\sim\chi_\beta^2=\Gamma(\alpha=\beta/2,\theta=2)$ that satisfies the Gamma distribution.  Writing $S_m\equiv\sum_nx_{mn}$, for $\beta=1$ one has $U_{mn}=g_{mn1}/\sqrt{S_m}$, while for $\beta=2$ one has
$U_{mn}=(g_{mn1}+\mathrm i g_{mn2})/\sqrt{S_m}$.  For $\beta=4$, with
$m,n$ understood as doublet labels, $a_{mn}=(g_{mn1}+\mathrm i g_{mn2})/\sqrt{S_m}$ and $b_{mn}=(g_{mn3}+\mathrm i g_{mn4})/\sqrt{S_m}$.  Consequently $q_{mn}=|U_{mn}|^2=x_{mn}/S_m$ for $\beta=1,2$, while
$q_{\bar m\bar n}=|a_{mn}|^2+|b_{mn}|^2=x_{mn}/S_m$ for $\beta=4$.
This construction describes one fixed squared overlap matrix row. Orthogonality constrains different rows, so they are not independent. It follows that
\begin{equation}
(q_{m1},\ldots,q_{mN_\beta})
\sim
(x_{m1}/S_m,\ldots,x_{mN_\beta}/S_m)
\sim
\operatorname{Dirichlet}(\alpha,\ldots,\alpha),
\qquad \alpha=\frac{\beta}{2}.
\label{eq:row_weight_dirichlet}
\end{equation}
where the Dirichlet means the standard Dirichlet distribution function.
The standard Dirichlet moments therefore give, for $i\ne j$,
\begin{equation}
\begin{aligned}
\mathbb E\,q_{mi}
&=\frac{\alpha}{N_\beta\alpha}=\frac{1}{N_\beta},\\[1ex]
M_\beta
&\equiv\mathbb E(q_{mi}q_{mj})
=\frac{\alpha^2}{(N_\beta\alpha)(N_\beta\alpha+1)}
=\frac{1}{N_\beta(N_\beta+2/\beta)},
\\[1ex]
\mathbb E\,q_{mi}^2
&=\frac{\alpha(\alpha+1)}{(N_\beta\alpha)(N_\beta\alpha+1)}
=\left(1+\frac{2}{\beta}\right)M_\beta,
\qquad
\mu_{\beta}\equiv\frac{\mathbb E\,q_{mi}^{2}}{M_\beta}
=1+\frac{2}{\beta}.
\end{aligned}
\label{eq:moment_ratios}
\end{equation}
For $\beta=1,2$, the diagonal moment $\mathbb E q_{mi}^2$ is precisely
$\mathbb E|U_{mi}|^4$.  For $\beta=4$, $w=2q$, so the same formula gives
the moments of the doublet weight without introducing a separate set of
moment definitions.

\subsection{Class dependence of the 2SP frame potential}
\label{app:2SP_class_covariant}

Eq.(7) of the main text
\begin{equation}
    \begin{split}
        &F^{(k)}_{\mathrm{2SP}}= 
        \sum_{\pi,\sigma \in S_k}^D\sum_{m_a,n_a=1}^D  \prod_{a=1}^k \Bigl|U^{(\bm{1})}_{m_a n_a} \Bigr|^2\,\Bigl|U^{(\bm{1})}_{m_{\pi(a)} n_{\sigma(a)}} \Bigr|^2 .
    \end{split}
\end{equation}
is a sum of positive overlap probabilities, so its permutation channels do not cancel. After reducing $(\pi,\sigma)$ to the relative permutation $\rho=\sigma\pi^{-1}$, each replica contains either $\mathbb E q_{mi}^2$ or
$\mathbb E(q_{mi}q_{mj})$. The Dyson-class dependence of these moments
determines the class dependence of the 2SP.

\begin{theorem}[Symmetry-class dependence of the 2SP]
\label{thm:2SP_all_classes}
Let $k$ be a fixed positive integer, and let the eigenframes be independent Haar draws from the Wigner--Dyson class $\beta\in\{1,2,4\}$. In the large-$D$ limit, the perfectly filtered 2SP frame potential satisfies
\begin{equation}
F^{(k)}_{\mathrm{2SP}}(\infty)
=\lambda_\beta^k\,k!
\sum_{j=0}^{k}\binom{k}{j}\,!(k-j)\,\mu_{\beta}^{\,j}
\left[1+O\!\left(\frac{k^{2}}{D}\right)\right],
\label{eq:2SP_class_estimate}
\end{equation}
where
\begin{equation}
(\lambda_\beta,\mu_\beta)=
\begin{cases}
\bigl(\frac{D}{D+2},3\bigr), & \beta=1,\\[1ex]
\bigl(\frac{D}{D+1},2\bigr), & \beta=2,\\[1ex]
\bigl(\frac{4D}{D+1},\frac{3}{2}\bigr), & \beta=4.
\end{cases}
\label{eq:class_parameters}
\end{equation}
For $\beta=4$, the time filter acts on Kramers-doublet labels.
\end{theorem}

\begin{proof}
The number of row or column labels is $N_1=N_2=D$ and $N_4=D/2$.  Relabeling a channel $(\pi,\sigma)$ by $r=\pi(a)$ gives
\begin{equation}
\prod_{a=1}^k (s_\beta q_{m_an_a})(s_\beta q_{m_{\pi(a)}n_{\sigma(a)}})
=s_\beta^{2k}\prod_{r=1}^k q_{m_rn_r}q_{m_rn_{\rho(r)}},
\qquad \rho\equiv\sigma\pi^{-1}.
\label{eq:2SP_row_identity}
\end{equation}
Thus the two factors belonging to replica $r$ lie in the same row. If
$\rho(r)=r$, they give $\mathbb E q_{mi}^2$. For distinct column labels and
$\rho(r)\ne r$, they give $\mathbb E(q_{mi}q_{mj})$ with $i\ne j$.

Define
\begin{equation}
\lambda_\beta\equiv s_\beta^2N_\beta^2M_\beta.
\end{equation}
Using $(N_1,N_2,N_4)=(D,D,D/2)$, $(s_1,s_2,s_4)=(1,1,2)$, and Eq.~\eqref{eq:moment_ratios} gives Eq.~\eqref{eq:class_parameters}.  Restricting to distinct row and column labels, a channel with $f=\operatorname{fix}(\rho)$ contributes
\begin{equation}
\begin{split}
\mathcal C_\rho^{\mathrm{dist}}
&=s_\beta^{2k}
\sum_{\{m,n\}}^{\mathrm{dist}}
\mathbb E\prod_{r=1}^k q_{m_rn_r}q_{m_rn_{\rho(r)}}\\
&=s_\beta^{2k}N_\beta^{2k}
M_\beta^{k-f}(\mu_\beta M_\beta)^f
\left[1+O\!\left(\frac{k^2}{D}\right)\right]\\
&=\bigl(s_\beta^2N_\beta^2M_\beta\bigr)^k\mu_\beta^f
\left[1+O\!\left(\frac{k^2}{D}\right)\right]\\
&=\lambda_\beta^k\mu_\beta^f
\left[1+O\!\left(\frac{k^2}{D}\right)\right].
\end{split}
\label{eq:2SP_factorization}
\end{equation}
The second line follows because the $k$ row labels and $k$ column labels give $N_\beta^{2k}$ choices.  Each of the $k-f$ non-fixed replicas contributes off-diagonal term $M_\beta= \mathbb E q_{mi} q_{mj}$, whereas each of the $f$ fixed replicas contributes diagonal term $\mathbb E q_{mi}^2=\mu_\beta M_\beta$.  Coincident labels and correlations between distinct squared overlap matrix account for the stated correction. 

One important observation is that both diagonal and off-diagonal terms contributes the leading order contribution, and therefore the result is symmetry class dependent.

Finally, for each relative permutation $\rho=\sigma\pi^{-1}$ there are $k!$ choices of $\pi$, after which $\sigma=\rho\pi$ is fixed.  Equation~\eqref{eq:2SP_factorization} thus gives
\begin{equation}
F^{(k)}_{\mathrm{2SP}}(\infty)
=\lambda_\beta^k\,k!
\sum_{\rho\in S_k}\mu_\beta^{\operatorname{fix}(\rho)}
\left[1+O\!\left(\frac{k^{2}}{D}\right)\right].
\label{eq:2SP_before_rencontres}
\end{equation}
There are $\binom{k}{j}\,!(k-j)$ permutations with exactly $j$ fixed points.  Hence
\begin{equation}
\sum_{\rho\in S_k}\mu^{\operatorname{fix}(\rho)}
=\sum_{j=0}^{k}\binom{k}{j}\,!(k-j)\,\mu^j,
\label{eq:rencontres}
\end{equation}
Substitution into Eq.~\eqref{eq:2SP_before_rencontres} proves Eq.~\eqref{eq:2SP_class_estimate}.
\end{proof}

The closed form in Eq.~\eqref{eq:rencontres} also gives
\begin{equation}
F^{(k)}_{\mathrm{2SP}}(\infty)
=\lambda_\beta^k(k!)^{2}
\sum_{\ell=0}^{k}\frac{(\mu_\beta-1)^{\ell}}{\ell!}
\left[1+O\!\left(\frac{k^{2}}{D}\right)\right],
\label{eq:2SP_closed_form}
\end{equation}
so that in the regime $1\ll k\ll\sqrt D$,
\begin{equation}
F^{(k)}_{\mathrm{2SP}}(\infty)\simeq
\lambda_\beta^k(k!)^{2}\,e^{\mu_\beta-1}.
\label{eq:2SP_asymptotic}
\end{equation}
Thus the 2SP is of order $(k!)^2$, rather than the Haar value $k!$.  Since $\lambda_4\to4$, the GSE has an additional asymptotic factor $4^k$.

\subsection{Class independence of the 3SP frame potential}
\label{app:3SP_class_blind}

The 3SP frame potential Eq.(11) of the main text 
\begin{equation}
\label{eq:3step_FP_perfect}
\begin{split}
    F^{(k)}_{\mathrm{3SP}}
=&\sum_{\pi,\sigma,\tau,\upsilon\in S_k}\sum_{\{m,p,f,g\}=1}^D
\prod_{a=1}^k
\Bigl(U^{(1)}_{m_a p_a}\,U^{(2)}_{p_a g_a}\,U^{*(2)}_{f_a g_a}\,U^{*(1)}_{m_a f_a}\Bigr)  \Bigl(U^{(1)}_{m_{\pi(a)} p_{\sigma(a)}}\,U^{(2)}_{p_{\sigma(a)} g_{\tau(a)}}\,U^{*(2)}_{f_{\upsilon(a)} g_{\tau(a)}}\,U^{*(1)}_{m_{\pi(a)} f_{\upsilon(a)}}\Bigr)^{*},
\end{split}
\end{equation} 
contains the loop quantity
$A_{mpgf}=U^{(1)}_{mp}\,U^{(2)}_{pg}\,U^{(2)*}_{fg}\,U^{(1)*}_{mf}$,
a product of overlap \emph{amplitudes}. Its leading class independence follows from row symmetry and the two moments used above.

\begin{theorem}[Class independence of the 3SP]
\label{thm:3SP_all_classes}
For fixed $k$ and independent Haar eigenframes in the orthogonal, unitary, or symplectic class,
\begin{equation}
F^{(k)}_{\mathrm{3SP}}(\infty)
=k!\left[1+O_{\beta}\!\left(\frac{k^2}{D}\right)\right]
=k!+O_{k,\beta}(D^{-1}).
\end{equation}
For $\beta=4$, the perfect time filter is understood as a projector onto distinct Kramers doublets.
\end{theorem}

\begin{proof}
After taking the complex conjugate of the second parenthesis in Eq.(11), the channel $(\pi,\sigma,\tau,\upsilon)$ is
\begin{equation}
\begin{split}
\mathcal C_{\pi\sigma\tau\upsilon}
=\sum_{\{m,p,f,g\}}
\mathbb E\prod_{a=1}^{k}&
U^{(1)}_{m_ap_a}\,U^{(1)*}_{m_af_a}\,
U^{(1)*}_{m_{\pi(a)}p_{\sigma(a)}}\,
U^{(1)}_{m_{\pi(a)}f_{\upsilon(a)}}
U^{(2)}_{p_ag_a}\,U^{(2)*}_{f_ag_a}\,
U^{(2)*}_{p_{\sigma(a)}g_{\tau(a)}}\,
U^{(2)}_{f_{\upsilon(a)}g_{\tau(a)}} .
\end{split}
\label{eq:3SP_class_channel}
\end{equation}
Restrict first to pairwise-distinct $m_a,g_a$ and to pairwise-distinct labels in $\{p_a,f_a\}_{a=1}^{k}$. The omitted collision sector is an $O(k^2/D)$ fraction of the sum. Row $m_r$ of $U^{(1)}$ then contains
\begin{equation*}
U^{(1)}_{m_rp_r},\quad U^{(1)*}_{m_rf_r},\quad
U^{(1)*}_{m_rp_{\sigma\pi^{-1}(r)}},\quad
U^{(1)}_{m_rf_{\upsilon\pi^{-1}(r)}}.
\end{equation*}
For a real matrix, each coordinate must occur an even number of times. For a complex matrix, each coordinate must occur equally often starred and unstarred. For a quaternionic matrix, the same condition applies to each block under unit-quaternion rotations. Since the labels are distinct, a nonzero leading average requires
$\sigma\pi^{-1}(r)=r
$ and $\upsilon\pi^{-1}(r)=r$ for every $r$,
so $\sigma=\upsilon=\pi$.  Row $p_r$ of $U^{(2)}$ then contains a component at $g_r$ and its conjugate at $g_{\tau\pi^{-1}(r)}$, which similarly requires $\tau=\pi$.  Hence only $\pi=\sigma=\tau=\upsilon$ survives outside the collision sector in every class.

Let $q^{(1)}$ and $q^{(2)}$ be the normalized row weights associated with $U^{(1)}$ and $U^{(2)}$.  After setting $\pi=\sigma=\tau=\upsilon$ and relabeling the replicas, the distinct-label part of an aligned channel is
\begin{equation}
\mathcal C_{\pi\pi\pi\pi}^{\mathrm{dist}}
=\sum_{\{m,p,f,g\}}^{\mathrm{dist}}
\mathbb E\prod_{a=1}^k
q^{(1)}_{m_ap_a}q^{(1)}_{m_af_a}
q^{(2)}_{p_ag_a}q^{(2)}_{f_ag_a}.
\label{eq:3SP_aligned_q}
\end{equation}
For $\beta=1,2$, Eq.~\eqref{eq:3SP_aligned_q} follows directly by taking the absolute square of the aligned amplitude loop.  For $\beta=4$, write each nonzero block as $U_{\bar m\bar p}=\sqrt{q_{\bar m\bar p}}R_{\bar m\bar p}$ with $R_{\bar m\bar p}\in SU(2)$.  The block orientations give $\mathbb E_R|\operatorname{tr}R|^2=1$, so the same expression holds after the orientation average, with barred block labels.

In an aligned channel the two independent overlap matrices can be averaged separately.  Since the distinct-label contribution for $p_a\ne f_a$ is the leading order contribution, which only involves off-diagonal term $\mathbb E(q_{mi}q_{mj})= M_\beta$. The final result is obtained directly from Eq.~\eqref{eq:3SP_aligned_q}:
\begin{equation}
\begin{aligned}
\mathcal C_{\pi\pi\pi\pi}
&=\bigl(N_\beta^2M_\beta\bigr)^{2k}
\left[1+O_\beta\!\left(\frac{k^2}{D}\right)\right]=1+O_\beta\!\left(\frac{k^2}{D}\right).
\end{aligned}
\end{equation}
Here $(N_1,N_2,N_4)=(D,D,D/2)$. The condition $p_a=f_a$ removes one free
summation, reducing the number of configurations by $1/N_\beta=O(1/D)$, where $\mu_{\beta}$ only enters in the subleading term. The trace of a
coincident GSE block also enters in the subleading term. The $k!$ choices of $\pi$ prove the theorem. The result is asymptotic, and the three classes don't need to agree at finite $D$.
\end{proof}

At $k=1$ the finite-size correction follows by separating $p\ne f$ and $p=f$.  For $\beta=1,2$,
\begin{equation}
\begin{split}
F^{(1)}_{\mathrm{3SP}}(\infty)
=\sum_{m,g}\Bigg[&
\sum_{\substack{p,f\\p\ne f}}
\mathbb E\!\left(q^{(1)}_{mp}q^{(1)}_{mf}\right)
\mathbb E\!\left(q^{(2)}_{pg}q^{(2)}_{fg}\right)+\sum_p
\mathbb E\!\left[(q^{(1)}_{mp})^2\right]
\mathbb E\!\left[(q^{(2)}_{pg})^2\right]\Bigg].
\end{split}
\label{eq:3SP_k1_split}
\end{equation}
There are $N_{\beta}^2$ choices of $(m,g)$, $N_{\beta}(N_{\beta}-1)$ ordered pairs $(p,f)$ with $p\ne f$, and $N_{\beta}$ diagonal choices $p=f$. 

Eq.~\eqref{eq:3SP_k1_split} gives
\begin{equation}
\begin{split}
F^{(1)}_{\mathrm{3SP}}(\infty)
&=N_{\beta}^2\left[N_{\beta}(N_{\beta}-1)M_\beta^2+N_{\beta}(\mu_\beta M_\beta)^2\right]=\begin{cases}
\dfrac{D(D+8)}{(D+2)^2},&\beta=1,\\[2ex]
\dfrac{D(D+3)}{(D+1)^2},&\beta=2.
\end{cases}
\end{split}
\label{eq:3SP_k1_beta12_exact}
\end{equation}
For $\beta=4$, the same counting uses $N_4=D/2$ block labels. The $p=f$ has an additional factor $4$, because its loop is $q^{(1)}_{\bar m\bar p}q^{(2)}_{\bar p\bar g}\mathbbm 1_2$ and its trace squared is $4(q^{(1)}_{\bar m\bar p})^2(q^{(2)}_{\bar p\bar g})^2$.  Therefore
\begin{equation}
\begin{aligned}
F^{(1)}_{\mathrm{3SP}}(\infty)\big|_{\beta=4}
&=N_4^2\left[N_4(N_4-1)M_4^2
+4N_4(\mu_4M_4)^2\right]\\
&=(N_4^2M_4)^2
\left[1+\frac{4\mu_4^2-1}{N_4}\right]\\
&
=\frac{D(D+16)}{(D+1)^2}.
\end{aligned}
\label{eq:3SP_k1_beta4_exact}
\end{equation}
Equations~\eqref{eq:3SP_k1_beta12_exact} and
\eqref{eq:3SP_k1_beta4_exact} approach $1$ with excesses $4/D$, $1/D$,
and $14/D$ for GOE, GUE, and GSE, respectively. At the dimensions used in Fig.(3) of the main text, the predicted residuals are $0.053$ at $D=70$,
$0.017$ at $D=56$, and $0.055$ at $D=252$. The corresponding measured
values are $0.056$, $0.021$, and $0.059$. Thus the symmetry class controls
the finite-size correction through $\mathbb E q_{mi}^2$ and
$\mathbb E(q_{mi}q_{mj})$, but not the common large-$D$ limit.
\hfill$\square$

\section{Appendix E: The experimental observables that distinguishes 2SP and 3SP}
In this appendix, we discuss several concrete experimental diagnostics that can distinguish the 2SP from the 3SP.

The most direct protocol is to measure an entangled return probability. Introduce a reference system $R$ with the same Hilbert-space dimension $D$ as the system $S$, and prepare
\begin{equation}
    |\Phi_D\rangle_{SR} =
\frac{1}{\sqrt D}
\sum_{j=1}^{D}|j\rangle_S|j\rangle_R.
\end{equation}
For two independently sampled unitaries $U,V$ from the temporal ensemble, apply $V$ followed by $U^\dagger$ to $S$. The probability of returning to $|\Phi_D\rangle$ is
\begin{equation}
    p_\Phi(U,V) =
\frac{|\Tr(U^\dagger V)|^2}{D^2}.
\end{equation}
Therefore, the $k$-th frame potential of this ensemble is given by
\begin{equation}
    F_{\mathcal E}^{(k)} =
D^{2k}\,
\mathbb E_{U,V\sim\mathcal E}
\left[p_\Phi(U,V)^k\right].
\end{equation}
Thus, repeated measurements of this return probability provide a direct experimental reconstruction of the frame potential. In our setting, this measurement would clearly distinguish the two protocols: the 2SP saturates above the Haar value, whereas the 3SP approaches $F_{\mathrm{Haar}}^{(k)}=k!$.

This protocol requires implementation of the inverse evolution and a projection onto the maximally entangled state, but it directly probes the quantity used throughout our theoretical analysis.

A complementary approach is to measure ensemble-averaged higher-order OTOCs. For example,
\begin{equation}
    C_{\mathbf A,\mathbf B}^{(k)}(U) =
\frac{1}{D}
\Tr\!\left[
A_1U^\dagger B_1U\cdots
A_kU^\dagger B_kU
\right].
\end{equation}
The complete set of such $2k$-point correlators determines the $k$-th moment channel and is directly related to the frame potential~\cite{Roberts_2017}. Higher-order OTOCs have also become experimentally accessible in quantum processors~\cite{GoogleQuantumAI2025ConstructiveInterference}. While complete certification would require many correlators, sampling a subset provides a practical diagnostic of the difference between the 2SP and 3SP.

Simpler witnesses can also be obtained from the output states. For a fixed input state $|\psi\rangle$, an exact unitary $k$-design reproduces Haar output-probability moments
\begin{equation}
    \mathbb E_U
|\langle x|U|\psi\rangle|^{2m} = \frac{m!}{D(D+1)\cdots(D+m-1)},
\qquad m\leq k,
\end{equation}
also the Haar values of subsystem quantities such as the quantum purity
\begin{equation}
    \mathbb E_{\rm Haar}\Tr\rho_A^2 =
\frac{d_A+d_B}{d_Ad_B+1}.
\end{equation}

These quantities can be accessed using randomized measurements \cite{Emerson_2005,Knill_2008} and Rényi-entropy protocols already implemented experimentally. They do not by themselves constitute complete certification of a unitary $k$-design, but they provide experimentally convenient witnesses of Haar-like behavior.


\end{document}